\documentclass[showpacs,amsmath,twocolumn,prl, superscriptaddress]{revtex4-1}
\usepackage{pifont}
\usepackage{graphicx,epsfig,subfigure,dsfont,amssymb,amsthm,amsfonts,amsbsy,mathrsfs,amscd}
\usepackage{epsfig}
\usepackage{times}
\usepackage{color}
\usepackage{physics}
\usepackage{framed}
\usepackage{hyperref}
\hypersetup{  colorlinks=true,           linkcolor=blue,              citecolor=red,            urlcolor=blue  }

\newtheorem{theorem}{Theorem}
\newtheorem*{theorem*}{Theorem}
\newtheorem{result}[theorem]{Result}

\newtheorem{lemma}{Lemma}
\def\qed{\leavevmode\unskip\penalty9999 \hbox{}\nobreak\hfill
     \quad\hbox{\leavevmode  \hbox to.77778em{%
               \hfil\vrule   \vbox to.675em%
               {\hrule width.6em\vfil\hrule}\vrule\hfil}}
     \par\vskip3pt}

\begin{document}

%\title{One-Shot Detection Limits of Quantum Radar}

\title{One-Shot Detection Limits of Quantum Illumination with Discrete Signals}

%\title{Ultimate Limits of Quantum Illumination}

%\title{Ultimate Detection Limits of Quantum Illumination}

%\title{When are Entangled Signals Needed for Quantum Illumination?}

%\title{Detection Limits in Quantum illumination}

%\title{Quantum illumination Limits in Detecting Objects in Very Noisy Environments}

%\title{Limits of Quantum Illumination in Detecting Objects in Very Noisy Environments}

%\title{Optimal Performance of Quantum Illumination for Detecting Objects with a Partial Reflectivity}

%\title{Optimality of Quantum Illumination for Detecting Objects with Arbitrary Reflectivity}

%\title{Optimal Quantum Illumination with Minimum Errors in State Discrimination}

%\title{Minimum Error Discrimination in Quantum Illumination}

\author{Man-Hong Yung}
\email{yung@sustc.edu.cn}
\affiliation{Institute for Quantum Science and Engineering and Department of Physics, Southern University of Science and Technology, Shenzhen 518055, China}
\affiliation{Shenzhen Key Laboratory of Quantum Science and Engineering, Shenzhen, 518055, China}

\author{Fei Meng}
\email{mengf@mail.sustc.edu.cn}
\affiliation{Department of Physics, Southern University of Science and Technology, Shenzhen 518055, China}
\affiliation{Department of Computer Science, The University of Hong Kong, Pokfulam Road, Hong Kong}

\author{Ming-Jing Zhao}
\email{zhaomingjingde@126.com}
\affiliation{School of Science, Beijing Information
Science and Technology University, Beijing, 100192, P. R. China}

\begin{abstract}
A minimally-invasive way to detect the presence of a stealth target is to probe it with a single photon and analyze the reflected signals. The efficiency of such a conventional detection scheme can potentially be enhanced by the method of quantum illumination, where entanglement is exploited to break the classical limits. The question is, what is the optimal quantum state that allows us to achieve the detection limit with a minimal error? Here we address this question for discrete signals, by deriving a complete and general set of analytic solutions for the whole parameter space, which can be classified into three distinct regions, in the form of ``phase diagrams" for both conventional and quantum illumination. Interestingly, whenever the reflectivity of the target is less than some critical values, all received signals become useless, which is true even if entangled resources are employed. However, there does exist a region where quantum illumination can provide advantages over conventional illumination; there, the optimal signal state is an entangled state with an entanglement spectrum inversely proportional to the spectrum of the environmental state. These results not only impose fundamental limits in applications such as quantum radars, but also suggest how to become immune against the attack of minimally-invasive detection.
\end{abstract}

\maketitle

%\section{Introduction}
{\bf Introduction---} One of the most important tasks in quantum information science is to understand how physical procedures related to information processing can be improved by exploiting quantum resources such as entanglement~\cite{Horodecki2009a}. Apart from the well-established applications such as quantum computation~\cite{kitaev2002classical}, simulation~\cite{Buluta2009a,Yung2012c}, teleportation~\cite{Bennett1998}, metrology~\cite{Giovannetti2006b}, etc., the area of quantum illumination~\cite{Lloyd2008b,Tan2008,Shapiro2009,Guha2009,Lopaeva2013,Barzanjeh2014,Zhang2014,Zhang2014a,Bradshaw2016,Sanz2016,Lanzagorta2016,Liu2017,Zhuang2017,Zhuang2017a,LasHeras2017} is emerging as a promising and novel quantum method for increasing the sensitivity or resolution of target detection in a way that can go beyond the classical limits. The primary goal of quantum illumination is to detect the presence or absence of a target, with potentially a low reflectivity and in a highly-noisy background, by sending out an entangled signal and performing joint (POVM) measurements. More specifically, the setup of quantum illumination consists of three parts: (i) a source emits a signal entangled with an idler system kept by an receiver; (ii) if a target exists, the receiver obtains the reflected part of the signal in addition to the background noise; otherwise, only the background noise can be received; (iii) the receiver perform a joint POVM measurement on the whole quantum system and infer from it the presence of the target. 

An intriguing feature of quantum illumination is that it is highly robust against loss and decoherence; one can still gain quantum advantages, even if the signal is applied to entanglement-breaking channels~\cite{Ruskai2003}.  As an important application, one can apply quantum illumination to secure quantum communication~\cite{Shapiro2009,Xu2011b,Zhang2013a,Ralph2013a,Shapiro2014}, where the sender encode a 0-or-1 message by controlling the presence of absence of an object and the receiver determine its presence by illuminating entangled photons; in this way, an eavesdropper who does not have access to another half of the entangled signal could virtually know nothing about the message communicated~\cite{Shapiro2009}. 

An experimental implementation~\cite{Zhang2013a} of the protocol above suggested that quantum illumination can provide a reduction up to five orders of magnitude in the bit-rate error against an eavesdropper attack. Furthermore, experimental implementations of quantum illumination have been extended from the optical domain~\cite{Zhang2013a,Lopaeva2013,Zhang2015} to the microwave domain~\cite{Barzanjeh2014}. This progress is significant, as in the optical domain, the natural (thermal) background radiation on average contains less than one photon per mode. Consequently, artificial noise is necessary to implement quantum illumination at optical wavelengths~\cite{Zhang2013a}. 

In fact, quantum illumination represents an applications of a larger class of problems called {\it quantum channel discrimination}~\cite{Harrow2010}. However, only very few analytic solutions have been discovered; in fact, quantum channel discrimination is generally a very hard computational problem~\cite{Rosgen2005}; it is complete for the quantum complexity class $\sf QIP$ (problems solvable by a quantum interactive proof system), which has been shown~\cite{Jain2010} to be equivalent to the complexity class $\sf PSPACE$ (problems solvable by classical computer with polynomial memory).

Here we show that the problem of one-shot quantum illumination, for any given parameter regime and for signals with any finite dimension, can be solved {\it completely} with a compact analytic solution. More specifically, our main results include a derivation of an analytic expression for the minimized error probability for target detection in quantum illumination, where the minimization is over all possible POVM measurements and for all possible finite-dimensional (entangled) probe states. Furthermore, the optimal state we obtained depends only on the spectral information of the environment signal; in other words, the minimized error probability can always be achieved without even knowing the reflectivity and occurrence probability of the target. 

On the other hand, quantum discord, a measure of non-classical correlation~\cite{Modi2012a}, was suggested~\cite{Weedbrook2013,Bradshaw2016} to be the reason for the quantum advantages gained by quantum illumination. However, this conclusion is not applicable to our results. In fact, the authors~\cite{Weedbrook2013} only consider completely-mixed environment; one can construct counter examples violating the conclusion of Ref.~\cite{Weedbrook2013} for general environments (see appendix).

{\bf Model of one-shot quantum illumination---} Let us first consider {\it conventional illumination}. Suppose the individual photonic state $\rho$ be described by an $d$-dimensional density matrix, and the thermal noise of the environment is denoted by, $\rho_E = \sum_{i=1} ^d \lambda_i \ket{\theta_i} \bra{\theta_i}$, where $\lambda_i$ satisfies $\lambda_1\geqslant\lambda_2...\geqslant\lambda_d$ and $\sum_{i=1}^{d}\lambda_i =1$. (i) if the target is absent, the probe signal $\rho$ is completely lost; we can only receive the noisy state from the environment, i.e., 
\begin{equation}
{{\cal E}_0}(\rho ) = {\rho _E} \ .
\end{equation}
(ii) even if the target is present, the detection may not be perfect; the reflecting portion of the signal is quantified by the {\it reflectivity}, $\eta \in [0,1]$, and the quantum channel is, 
\begin{equation}
{{\cal E}_1}(\rho ) = \eta \rho  + (1 - \eta ){\rho _E} \ .
\end{equation}
For {\it quantum illumination}, the probe signal is entangled with another subsystem, and the quantum channels are applied partially, i.e., (i) when the target is absent, $({\cal E}_0 \otimes \mathcal{I})(\rho_{AB})=\rho_E\otimes \rho_B$, where $\rho_B \equiv {\rm tr}_A \rho_{AB}$, and when it is present, $({\cal E}_1\otimes\mathcal{I})(\rho_{AB})=\eta\rho_{AB}+(1-\eta) \rho_E\otimes \rho_B $. 

\begin{figure}[t!]
    \centering
\includegraphics[width=0.9\columnwidth]{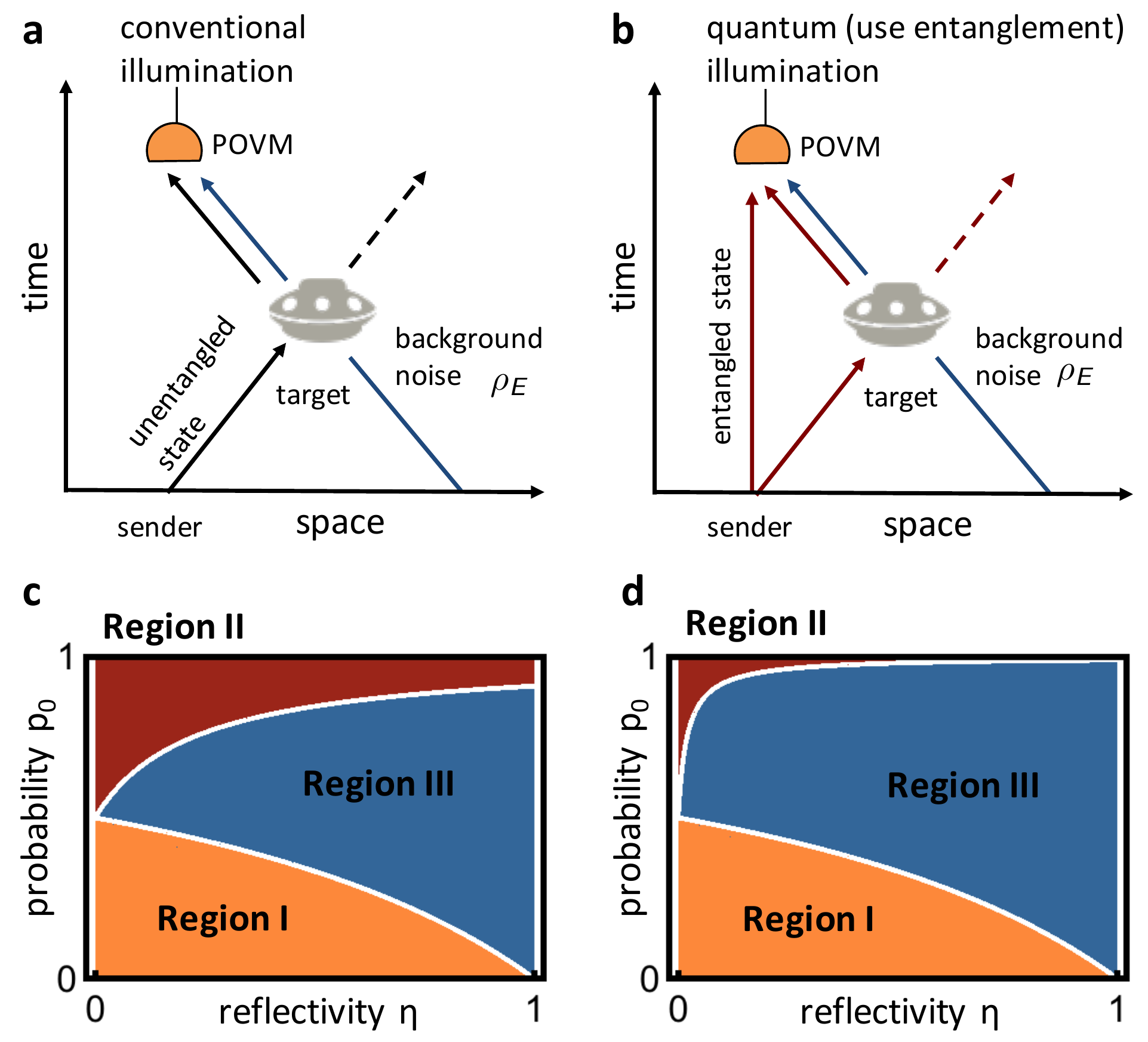}
    \caption{Conventional and quantum illumination. (a) In conventional illumination, a signal is sent to probe a target without the use of entanglement. (b) In quantum illumination, the signal is entangled, and a join POVM measurement is performed at the end to reduce the detection error. (c) The phase diagram for conventional illumination and (d) same diagram for quantum illumination.}
    \label{fig:quan_vs_class}
\end{figure}

The problem of target detection with quantum illumination can be regarded as a problem of {\it quantum channel discrimination}~\cite{Harrow2010}: given a pair of quantum states $\rho_0$ and $\rho_1$, associated with probabilities $p_0$ and $p_1$, $p_0+p_1=1$.  In connection with the problem of quantum illumination, we should take $\rho_0 = {\cal E}_0(\rho)$ and $\rho_1 = {\cal E}_1(\rho)$ for conventional illumination, and $\rho_0 = ({\cal E}_0 \otimes \mathcal{I})(\rho_{AB})$ and $\rho_1 = ({\cal E}_1\otimes\mathcal{I})(\rho_{AB})$ for quantum illumination. We are interested in finding the input states that can minimize the minimum error, given by~\cite{Helstrom1969},
\begin{equation}
{P_{{\text{err}}}} \equiv \frac{1}{2}\left( {1 - \left\| {{p_0}{\rho _0} - {p_1}{\rho _1}} \right\|} \right) \ ,
\end{equation}
where $\left\| A \right\| \equiv {\text{tr}}\sqrt {{A^\dag }A}$ denotes the trace norm of a matrix $A$. Here, the corresponding trace norms are labeled by $\left\| {{\Omega _{{c}}}(\rho )} \right\|$ and $\left\| {{\Omega_{ q}}({\rho _{AB}})} \right\|$, respectively for conventional and quantum illumination, where 
\begin{equation}
{\Omega _{{c}}}(\rho ) \equiv {p_1}{{\cal E}_1}\left( \rho  \right) - {p_0}{{\cal E}_0}\left( \rho  \right) \ ,
\end{equation}
 and 
\begin{equation}
{\Omega _{{q}}}(\rho_{AB} ) \equiv {p_1}({{\cal E}_1} \otimes \mathcal{I})({\rho _{AB}}) - {p_0}({{\cal E}_0} \otimes \mathcal{I})({\rho _{AB}}) \ .
\end{equation}
In other words, the corresponding minimum errors are given by $  P_{{{\rm err}}}^{ c,q} \equiv \frac{1}{2}\left( {1 - \left\| {{\Omega _{c,q}}\left( \rho  \right)} \right\|} \right)$. Our ultimate task is to optimize over all possible states, i.e., $\left\| {{\Omega _c}} \right\|_\diamond \equiv \mathop {\max }\nolimits_\rho  \left\| {{p_0}{{\cal E}_0}\left( \rho  \right) - {p_1}{{\cal E}_1}\left( \rho  \right)} \right\|$ for conventional illumination and ${\left\| {{\Omega _q}} \right\|_\diamond } \equiv \mathop {\max }\limits_{{\rho _{AB}}} \left\| {{p_1}({{\cal E}_1} \otimes \mathcal{I})({\rho _{AB}}) - {p_0}({{\cal E}_0} \otimes \mathcal{I})({\rho _{AB}})} \right\|$ for quantum illumination, where $P_{{\rm{err}},\diamondsuit }^c \equiv \mathop {\max }\limits_{{\rho _{AB}}} P_{{\rm{err}}}^c$ and $P_{{\rm{err}},\diamondsuit }^q \equiv \mathop {\max }\limits_{{\rho _{AB}}} P_{{\rm{err}}}^q$, or explicitly,
\begin{equation}
  P_{{\rm err}, \diamond }^c \equiv \frac{1}{2}( {1 - {{\left\| {{\Omega _c}} \right\|}_\diamond }} ) \ \ {\rm and} \ \ P_{{\rm err},\diamond}^q \equiv \frac{1}{2}( {1 - {{\left\| {{\Omega _q}} \right\|}_\diamond}} ) \ .
\end{equation}
When the two subsystems are uncorrelated, i.e., ${\rho _{AB}} \equiv {\rho _A} \otimes {\rho _B}$, the quantum case is reduced to the conventional case; therefore, it is necessarily true that quantum illumination is not worse than conventional illumination, i.e., $ P_{{\text{err}}}^q \ \leqslant \ P_{{\text{err}}}^c$ and $P_{{{\rm err},\diamond}}^q \ \leqslant \ P_{{{\rm err},\diamond}}^c$. Finally, we note that both the values of $p_0$ (and $p_1$), and the reflectivity $\eta$ can be determined in the beginning by state tomography.

{\bf Main results---}
Our major results contain a family of {\it complete} analytic solutions for both conventional and quantum illumination for any $d$-dimensional signal state and any given environmental state~$\rho_E$. For both conventional and quantum illuminations, the minimal-error probabilities are strongly dependent on the the occurrence probabilities $\{ p_0, p_1 \}$ and the reflectivity $\eta$ of the target. In general, we can divide the parameter space into three distinct regions, namely (I,II,III).

{\sf (Region I):} (i) $p_0 < p_1$,  and (ii) $\eta  < {\eta _*} \equiv 1 - {p_0}/{p_1}$. For both conventional and quantum illuminations, the minimal error is given by, 
\begin{equation}
P_{\rm err} = p_0 \ .
\end{equation}
Furthermore, the optimal strategy for both quantum and conventional illumination {\it does not even require a measurement of the signals}; one can simply guess ``yes" (present of the target) for all cases. As whenever $p_0 < p_1$, the error for this simple strategy is equal to $p_0$, i.e., $P_{\rm err} = p_0$. We summary this result as follows (the proof is left in the appendix):
\begin{result}[\bf Region I, for both conventional and quantum illuminations]\label{theo:con_pp0l12} (a) the minimal errors for conventional and quantum illumination are equal to $p_0$, i.e., ${P}_{{\rm{err}}} = p_0$, and (b) the bound can be achieved with any (pure or mixed) state.
\end{result}

{\sf (Region II):} (i) $p_0 > p_1$, and (iia) for conventional illumination: $\eta < \eta_c  \equiv (\tfrac{{{p_0}}}{{{p_1}}} - 1)(\tfrac{{{\lambda _d}}}{{1 - {\lambda _d}}})$, or (iib) for quantum illumination: $\eta  < {\eta _q} \equiv ( {\tfrac{{{p_0}}}{{{p_1}}} - 1} ) ( {\frac{{{\lambda _h}}}{{1 - {\lambda _h}}}} )$. The minimal error is given by, 
\begin{equation}
P_{\rm err} = p_1 \ ,
\end{equation}
for both conventional and quantum illuminations.  Moreover, both ${\eta _c} \to 0$ and ${\eta_q} \to 0$ vanishes as $\lambda_d \to 0$, which implies that region II vanishes for both cnventional and quantum illuminations. The same performance is achieved by guessing ``no" (i.e., absence of the target) for all events. Here 
\begin{equation}
{\lambda_h^{-1}} \equiv {\sum\limits_{i = 1}^d {\lambda _i^{ - 1}} } \ ,
\end{equation}
which is related to the harmonic mean of of the eigenvalues $\{ \lambda_i \}$ of environmental signal $\rho_E$. Note that for $\lambda_d > 0$, it is always true that $\lambda_h$ is always less than the smallest eigenvalue $\lambda_d$ of $\rho_E$, i.e., 
\begin{equation}
{\lambda _h} < {\lambda _d} \ ,
\end{equation}
 (because ${\lambda _d}/{\lambda _h} = {\lambda _d}(\sum\nolimits_{i = 1}^d {1/{\lambda _i}} ) > 1$). Therefore, the region II for the case of quantum illumination is { always} {\it smaller} than that of conventional illumination. (see Fig.~\ref{fig:quan_vs_class}). To summarize (see proof in appendix), we have
\begin{result}[\bf Region II for conventional and quantum illumination]\label{theo_con_p0g12}
(a) the minimal error for conventional and quantum illumination is equal to $p_1$, i.e., ${P}_{{\rm{err}}} = p_1$, and (b) the bound can be achieved with any (pure or mixed) state.
\end{result}

{\sf (Region III):} (the region excluded by region I and II) {For conventional illumination, the minimal error over all possible input states is given by, 
\begin{equation}
P_{{{\rm err},\diamond}}^c = {p_0} + \gamma \left( {1 - {\lambda _d}} \right) \ ,
\end{equation}
and for quantum illumination,
\begin{equation}\label{RegionIII_P_err}
P_{{\text{err},\diamond}}^q = {p_0} + \gamma \left( {1 - {\lambda_h}} \right) \ .
\end{equation}
} Here the parameter, $\gamma  \equiv {p_1}(1 - \eta ) - {p_0}$, depends on the occurrence probabilities $\{ p_0, p_1 \}$ of the target and the reflectivity~$\eta$. In this region, $\gamma <0$ is negative and a decreasing function of $\eta$, which implies that both $P_{{{\rm err},\diamond}}^c$ and $P_{{{\rm err},\diamond}}^c$ decrease with the increase of the reflectivity~$\eta$. Furthermore, the difference between the classical and quantum cases (i.e., quantum advantage) depends on the difference, ${{\lambda _d} - \lambda_h }$ , i.e.,
\begin{equation}
P_{{\text{err}},\diamond }^c - P_{{\text{err}},\diamond }^q = \left| \gamma  \right|\left( {{\lambda _d} - {\lambda _h}} \right) \geq 0 \ .
\end{equation}
 For conventional illumination, the input state that can minimize the detection error is given by the eigenstate $\ket{\theta_d}$ of $\rho_E$ associated with the smallest eigenvalue $\lambda_d$. To summarize (see proof in appendix), we have
\begin{result}[\bf Region III: minimal error decreases with reflectivity $\eta$  for conventional illumination]\label{RegionIII_class_Perr}
The minimal error $P_{\rm err}$	over all possible conventional input states is given by $P_{{\rm{err}}}^c = {p_0} + \gamma \left( {1 - {\lambda _d}} \right)$, which is obtained by choosing $\left| \psi  \right\rangle  = \left| {{\theta_d}} \right\rangle $ to be the eigenvector of $\rho_E$ associated with the smallest eigenvalue.
\end{result}

To understand the result (Eq.~(\ref{RegionIII_P_err})) of the minimal error for quantum illumination in Region III, we summarize the steps for achieving it below:

{\bf Sketch of the proofs for quantum illumination---} The main physical quantity to be investigated is: $  {\Omega _q}\left( {{\rho _{AB}}} \right) = {p_1}\eta {\rho _{AB}} + \gamma {\rho _E} \otimes {\rho _B}$. For convenience, we can focus on the following matrix, ${H_q} \equiv - {\Omega _q}/ \gamma $, where 
\begin{equation}
H_q = {\rho _E} \otimes {\rho _B} - \alpha \left| \psi  \right\rangle \left\langle \psi  \right| \ ,
\end{equation}
and $\alpha  =  - \eta {p_1}/\gamma  = \eta {p_1}/\left( {{p_0} - {p_1}\left( {1 - \eta } \right)} \right) > 0$. Furthermore, we express the bipartite pure state in the following form: $\left| \psi  \right\rangle  = \sum\nolimits_{i = 1}^d {\left| {{\theta _i}} \right\rangle \left| {{u_i}} \right\rangle }$. where the vectors $ \ket{u_i}$'s are {\it not} assumed to be normalized. In general, they are non-orthogonal to one another. However, since the eigenvalues $\left| {{\theta _i}} \right\rangle$'s are orthonormal, the normalization condition implies that $\left\langle \psi  \right.\left| \psi  \right\rangle  = \sum\nolimits_{i = 1}^d {\left\langle {{u_i}} \right.\left| {{u_i}} \right\rangle }  = 1$.  

The next task is to bound the minimum eigenvalue, $E_g \equiv \lambda_{\rm min} (H_q)$ of the matrix $H_q$. The corresponding eigenvector $\left| {{g_{\psi}}} \right\rangle $, where ${H_q}\left| {{g_\psi }} \right\rangle  = {E_g}\left| {{g_\psi }} \right\rangle $, can always be expanded by the the following vectors (in a way similar to $\ket{\psi}$), $\left| {{g_\psi }} \right\rangle  = \sum\nolimits_{i = 1}^d {\left| {{\theta _i}} \right\rangle \left| {{v_i}} \right\rangle } $, where, again, the vectors $\left| {{v_i}} \right\rangle$'s are neither normalized nor orthogonal to one another. We found that (see appendix) the eigenvalue is minimized when we choose $\left| {{u_i}} \right\rangle  = \left| {{v_i}} \right\rangle$ for all $i$'s, which means that $\left| {{g_\psi }} \right\rangle  = \left| \psi  \right\rangle$. Finally, we found that the minimum eigenvalue, $E_g = \lambda_h - \alpha$, of $H_q$ can be achieved by choosing an input of the form (see appendix), 
\begin{equation}
\left| \psi  \right\rangle  = \sum\limits_{i = 1}^d {{\mu _i}} \left| {{\theta _i}} \right\rangle \left| {{\theta _i}} \right\rangle \ ,
\end{equation}
where ${\mu _i} = \sqrt {{\lambda _h}/{\lambda _i}} $; this result is summarized as follows:
\begin{result}[\bf Optimal state for quantum illumination]
The lower bound, $\lambda_h - \alpha$, of $E_g$ can be achieved 	by the input state, $  \left| \psi  \right\rangle  = \sum\nolimits_{i = 1}^d {{\mu _i}} \left| {{\theta _i}} \right\rangle \left| {{\theta _i}} \right\rangle$, where ${\mu _i} = \sqrt {{\lambda _h}/{\lambda _i}} $.
\end{result}

Below, we provide three different examples to illustrate our results.

{\bf Example 1: when $\lambda_{\rm min} = 0$.} For the eigenstate $\left| \psi  \right\rangle$ of ${\Omega _{{c}}}(\left| \psi  \right\rangle \langle \psi |) = {p_1}\eta \left| \psi  \right\rangle \langle \psi | + \gamma \, {\rho _E}$, i.e., $\left\langle \psi  \right|{\rho _E}\left| \psi  \right\rangle  = 0$. The error probability is ${P_{{\text{err}}}} = \frac{1}{2}[1 - \left( {{p_1}\eta  + \left| \gamma  \right|} \right)]$, which means that (i) when $\eta \le \eta_*$ (or $\gamma \ge 0$), then $P_{\rm err} = p_0$, and (ii) when  $\eta \ge \eta_*$ (or $\gamma \le 0$), then ${P_{\rm err}} = {p_1}\left( {1 - \eta } \right)$, which vanishes as expected when $\eta \to 1$.

{\bf Example 2: binary signals} Let us consider the case where the signals are two-dimensional, which means that $\rho_E$ is a $2 \times 2$ Hermitian matrix. In its diagonal basis (labeled as $\left\{ {\left| 0 \right\rangle ,\left| 1 \right\rangle } \right\}$), we write $ \rho_E=\bigl[\begin{smallmatrix}
\lambda_0 & 0 \\ 0 & \lambda_1
\end{smallmatrix} \bigr]$.
%\begin{equation}
%  {\rho _E} \equiv \left[ {\begin{array}{*{20}{c}}
%  {{\lambda _0}}&0 \\
%  0&{{\lambda _1}}
%\end{array}} \right] \ .
%\end{equation}
Since the trace norm is invariant under unitary transformation, we can always choose to have the pure state to be optimized as follows: $\left| \psi  \right\rangle  = {\mu _0}\left| 0 \right\rangle  + {\mu _1}\left| 1 \right\rangle$.
%\begin{equation}
%\left| \psi  \right\rangle  = {\mu _0}\left| 0 \right\rangle  + {\mu _1}\left| 1 \right\rangle   \ ,
%\end{equation}
where both parameters, ${\mu _0} \ge 0 $ and ${\mu _1} \ge 0$, are positive, and $\mu _0^2 + \mu _1^2 = 1$. Consequently, we have $ {\Omega _{{c}}} =\bigl[\begin{smallmatrix}
a & c \\ c & b
\end{smallmatrix} \bigr]$,
%\begin{equation}
%  {\Omega _{\text{c}}} = \left[ {\begin{array}{*{20}{c}}
%  a&c \\
%  c&b
%\end{array}} \right] \ ,
%\end{equation}
where $a \equiv {p_1}\eta \mu _0^2 + \gamma {\lambda _0}$, $b \equiv {p_1}\eta \mu _1^2 + \gamma {\lambda _1}$, and $c \equiv {p_1}\eta {\mu _0}{\mu _1}$.
The eigenvalues ${\lambda _ \pm }$ of ${\Omega _{{c}}}$ are given by ${\lambda _ \pm } = \frac{1}{2}\left[ {{\text{tr}}{\Omega _{{c}}} \pm \sqrt {{\text{t}}{{\text{r}}^2}{\Omega _{{c}}} - 4\det {\Omega _{{c}}}} } \right]$,
%\begin{equation}
%{\lambda _ \pm } = \frac{1}{2}\left[ {{\text{tr}}{\Omega _{\text{c}}} \pm \sqrt {{\text{t}}{{\text{r}}^2}{\Omega _{\text{c}}} - 4\det {\Omega _{\text{c}}}} } \right]  \ ,
%\end{equation}
where trace and determinant of $\Omega_{ c}$ are $ {\text{tr}} \ {\Omega _{{c}}} = a + b = {p_1}\eta  + \gamma = p_1 -p_0$, $\det {\Omega _{{c}}} = ab - {c^2} = {\gamma ^2}{\lambda _0}{\lambda _1} + {p_1}\eta \gamma \left( {{\lambda _0}\mu _1^2 + {\lambda _1}\mu _0^2} \right)$. The trace norm of ${\Omega _{{c}}}$ is given by one of the following possibility:
\begin{equation}
\left\| {{\Omega _{{c}}}} \right\| = \left\{ {\begin{array}{*{20}{c}}
  {\left| {{\text{tr}} \ {\Omega _{{c}}}} \right|}&{{\text{if }}\det {\Omega _{\text{c}}} \geqslant 0} , \\
  {\sqrt {{\text{t}}{{\text{r}}^2}{\Omega _{{c}}} - 4\det {\Omega _{{c}}}} }&{{\text{if }}\det {\Omega _{{c}}} < 0} .
\end{array}} \right.
\end{equation}
Note that $\det \Omega_c$ is a product of the two eigenvalues; the condition of $\det {\Omega _{\text{c}}} > 0$ implies that either both eigenvalues are positive or both negative.

{\bf Example 3: completely-mixed environment} Suppose the returning signal from the noisy environment is completely mixed, i.e., $\rho_E=I/d$, the corresponding matrix ${\Omega _{{c}}}(\left| \psi  \right\rangle \langle \psi |)$, for any pure state $\ket{\psi}$, can be diagonalized explicitly to give $\left\| {{\Omega _{{c}}}(\left| \psi  \right\rangle \langle \psi |)} \right\| = \left| {{p_1}\eta  + \frac{\gamma }{d}} \right| + \frac{{d - 1}}{d}\left| \gamma  \right|$. In region I, where $p_0\leq\frac{1}{2}$ and $\eta \le \eta_* \equiv 1 - p_0/p_1$, we have $\gamma \equiv {p_1}(1 - \eta ) - {p_0} \geq 0$. As a result, $\left\| {{\Omega _c}} \right\| = |{p_1}\eta  + \gamma|  = \left| {{p_1} - {p_0}} \right|$ and hence $P_{\rm err}^c = p_0$. On the other hand, in region~II, $p_0>\frac{1}{2}$ (where $\gamma < 0$) and $\eta  < ( {\frac{{{p_0}}}{{{p_1}}} - 1} )\frac{1}{{d - 1}}$ (where ${p_1}\eta  + \gamma /d < 0$), we have again $\left\| {{\Omega _c}} \right\| = \left| {{p_1}\eta  + \gamma } \right| = \left| {{p_1} - {p_0}} \right|$, but it gives 
$P_{\rm err}^c = p_1$. In region III, $\left\| {{\Omega _c}} \right\| =  {{p_1}\eta  + \left( {2/d - 1} \right)\gamma } $, which gives $P_{{\rm{err}}}^c = {p_1}\left( {1 - \eta } \right) - \gamma /d$.

\begin{figure}[t!]
    \centering
\includegraphics[width=0.9\columnwidth]{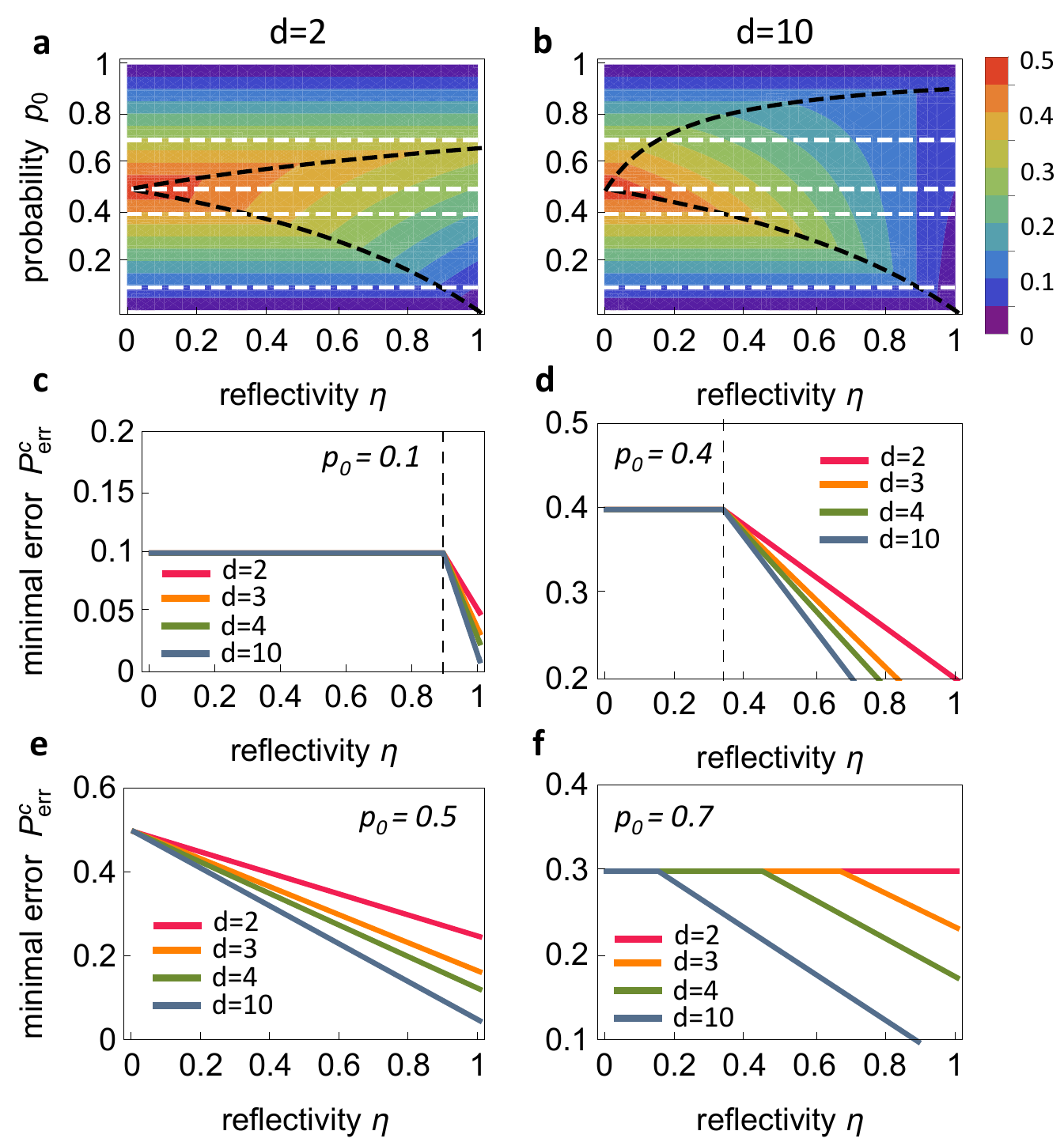}
    \caption{Behavior of conventional illumination in a completely-mixed environment $\rho_E=I/d$. (a) Color plot of the minimal-error $P_{\rm err}^c$ as a function of the occurence probability $p_0$ and reflectivity $\eta$ for two-dimensional signals $d=2$. (b) the same plot for $d=10$. Explicit dependence of $P_{\rm err}^c$ as a function of $\eta$ are shown in (c)-(f).}
    \label{fig:classicalIllu}
\end{figure}

{\bf Complementary results---} 
If for all the events, we simply guess `{\sf yes}' (i.e., presence of the target) whenever $p_0 < p_1$ and `{\sf no}' (i.e., absent) whenever $p_0 > p_1$. Then, the error for guessing wrong is  given by the minimum of the probabilities $p_0$ or $p_1$, i.e., $\min \left\{ {{p_0},{p_1}} \right\}$. For example, the instance shown below shows that the number of wrong decisions (i.e., `{\sf yes}' when the object is absent `0') is equal to the number of absent events `0'.
\begin{equation*}
\boxed{\begin{array}{*{20}{c}}
  0 \\ 
  {\sf yes} 
\end{array}}\begin{array}{*{20}{c}}
  1 \\ 
  {\sf yes} 
\end{array}\boxed{\begin{array}{*{20}{c}}
  0 \\ 
  {\sf yes} 
\end{array}}\begin{array}{*{20}{c}}
  1 \\ 
  {\sf yes} 
\end{array}\boxed{\begin{array}{*{20}{c}}
  0 \\ 
  {\sf yes} 
\end{array}}\begin{array}{*{20}{c}}
  1 \\ 
  {\sf yes} 
\end{array}\boxed{\begin{array}{*{20}{c}}
  0 \\ 
  {\sf yes} 
\end{array}}\begin{array}{*{20}{c}}
  1 \\ 
  {\sf yes} 
\end{array}\begin{array}{*{20}{c}}
  1 \\ 
  {\sf yes} 
\end{array}
\end{equation*}

This argument can be justified more formally for both classical and quantum illuminations with the following, where the proof is given in the appendix, 
\begin{result}[\bf Upper bound of minimal error]
The error probability $P_{\rm err}$ is bounded above by either $p_0$ or $p_1$, i.e., ${P_{{\rm{err}}}} \leqslant \frac{1}{2}\left( {1 - \left| {{p_0} - {p_1}} \right|} \right) = \min \left\{ {{p_0},{p_1}} \right\} $.
\end{result}
This error bound is relevant to the cases where the reflectivity $\eta$ is zero, i.e., ${\rho _0} = {\rho _1}$ for both conventional and quantum illumination, which gives ${P_{\rm err}} = \tfrac{1}{2}\left( {1 - |{p_0} - {p_1}|} \right)$. In other words, when there is no signal related to the absence/presence of the target, the best strategy one can make to minimize the error of discrimination is exactly the strategy mentioned above.

Another interesting question is how the reflectivity $\eta$ affect the error bound. Intuitively, we would believe that the higher the value of $\eta$, the smaller the error bound. This intuition can be justified by the following theorem (proof in the appendix):
\begin{result}[\bf Monotonicity of minimal error]
For a given reflectivity $\eta$, and density matrix $\rho$, and the minimal error given by, ${P_{err}}\left(\rm  \eta  \right) = \tfrac{1}{2}\left( {1 - \left\| {{p_1}{\rho _1}\left( \eta  \right) - {p_0}{\rho _0}} \right\|} \right)$, where ${\rho _1}\left( \eta  \right) = \eta \rho  + \left( {1 - \eta } \right){\rho _0}$. The minimal error is a non-increasing function of the reflectivity, i.e., if $\eta  \geqslant \eta '$, then $  {P_{\rm err}}\left( \eta  \right) \leqslant {P_{\rm err}}\left( {\eta '} \right)$.
\end{result}

On the other hand, the optimization can be taken over pure states only (see proof in appendix).
\begin{result}[\bf Optimization over pure states]\label{theo:optim_purestate}
In optimizing the trace norms of $\Omega_{c} (\rho)$ or $\Omega_{q} (\rho)$, over all possible density matrices $\rho$, $\mathop {\max }\nolimits_{\rho  \in \mathcal{H}} \left\| {{ \, \Omega _{\text{c}}}  (\rho )} \, \right\| = \mathop {\max }\nolimits_{|\psi \rangle  \in \mathcal{H}} \left\| {{ \, \Omega _{\text{c}} }\, (\left| \psi  \right\rangle \left\langle \psi  \right|)} \, \right\| $, and for quantum illumination, $\mathop {\max }\nolimits_{\rho_{AB}  \in \mathcal{H}\otimes\mathcal{H}} \left\| {{ \, \Omega _{\text{q}}}  (\rho_{AB} )} \, \right\| = \mathop {\max }\nolimits_{|\psi_{AB} \rangle  \in \mathcal{H}\otimes\mathcal{H}} \left\| {{ \, \Omega _{\text{q}} }\, (\left| \psi_{AB}  \right\rangle \left\langle \psi_{AB}  \right|)} \, \right\|$.
\end{result}

%%%%%%%%%%%%%%%%%%%%%%%%%%%%%%
{\bf Conclusions---}
In this work, we presented complete solutions to the problem of one-shot minimum-error discrimination for both conventional and quantum illuminations, for finite-dimensional signals. The analysis is divided into three regions. Region I are the same for both conventional and quantum illumination; the minimal error is a constant and does not depend on the reflectivity of the target，the optimal strategy is achieved via simple guess. The same is similar for region II, except that using quantum illumination can shrink the boundary of region II. For region III, quantum illumination can yield a lower minimal error than conventional illumination. 

\section{Acknowledgement}
We thank Cheng Guo and Mile Gu for useful discussions. This work is supported by the NSF of China under Grant No. 11401032 and Scientific Research Foundation for the Returned Overseas Chinese Scholars, State Education Ministry, the Guangdong Innovative and Entrepreneurial Research Team Program (No. 2016ZT06D348), and the Science Technology and Innovation Commission of Shenzhen Municipality (ZDSYS20170303165926217, JCYJ20170412152620376). M.-H.Y acknowledges support from the National Youth 1000 Talents Plan.

\bibliography{references}

\begin{thebibliography}{33}%
\makeatletter
\providecommand \@ifxundefined [1]{%
 \@ifx{#1\undefined}
}%
\providecommand \@ifnum [1]{%
 \ifnum #1\expandafter \@firstoftwo
 \else \expandafter \@secondoftwo
 \fi
}%
\providecommand \@ifx [1]{%
 \ifx #1\expandafter \@firstoftwo
 \else \expandafter \@secondoftwo
 \fi
}%
\providecommand \natexlab [1]{#1}%
\providecommand \enquote  [1]{``#1''}%
\providecommand \bibnamefont  [1]{#1}%
\providecommand \bibfnamefont [1]{#1}%
\providecommand \citenamefont [1]{#1}%
\providecommand \href@noop [0]{\@secondoftwo}%
\providecommand \href [0]{\begingroup \@sanitize@url \@href}%
\providecommand \@href[1]{\@@startlink{#1}\@@href}%
\providecommand \@@href[1]{\endgroup#1\@@endlink}%
\providecommand \@sanitize@url [0]{\catcode `\\12\catcode `\$12\catcode
  `\&12\catcode `\#12\catcode `\^12\catcode `\_12\catcode `\%12\relax}%
\providecommand \@@startlink[1]{}%
\providecommand \@@endlink[0]{}%
\providecommand \url  [0]{\begingroup\@sanitize@url \@url }%
\providecommand \@url [1]{\endgroup\@href {#1}{\urlprefix }}%
\providecommand \urlprefix  [0]{URL }%
\providecommand \Eprint [0]{\href }%
\providecommand \doibase [0]{http://dx.doi.org/}%
\providecommand \selectlanguage [0]{\@gobble}%
\providecommand \bibinfo  [0]{\@secondoftwo}%
\providecommand \bibfield  [0]{\@secondoftwo}%
\providecommand \translation [1]{[#1]}%
\providecommand \BibitemOpen [0]{}%
\providecommand \bibitemStop [0]{}%
\providecommand \bibitemNoStop [0]{.\EOS\space}%
\providecommand \EOS [0]{\spacefactor3000\relax}%
\providecommand \BibitemShut  [1]{\csname bibitem#1\endcsname}%
\let\auto@bib@innerbib\@empty
%</preamble>
\bibitem [{\citenamefont {Horodecki}\ \emph {et~al.}(2009)\citenamefont
  {Horodecki}, \citenamefont {Horodecki}, \citenamefont {Horodecki},\ and\
  \citenamefont {Horodecki}}]{Horodecki2009a}%
  \BibitemOpen
  \bibfield  {author} {\bibinfo {author} {\bibfnamefont {R.}~\bibnamefont
  {Horodecki}}, \bibinfo {author} {\bibfnamefont {P.}~\bibnamefont
  {Horodecki}}, \bibinfo {author} {\bibfnamefont {M.}~\bibnamefont
  {Horodecki}}, \ and\ \bibinfo {author} {\bibfnamefont {K.}~\bibnamefont
  {Horodecki}},\ }\href {\doibase 10.1103/RevModPhys.81.865} {\bibfield
  {journal} {\bibinfo  {journal} {Rev. Mod. Phys.}\ }\textbf {\bibinfo {volume}
  {81}},\ \bibinfo {pages} {865} (\bibinfo {year} {2009})}\BibitemShut
  {NoStop}%
\bibitem [{\citenamefont {Kitaev}\ \emph {et~al.}(2002)\citenamefont {Kitaev},
  \citenamefont {Shen},\ and\ \citenamefont {Vyalyi}}]{kitaev2002classical}%
  \BibitemOpen
  \bibfield  {author} {\bibinfo {author} {\bibfnamefont {A.~Y.}\ \bibnamefont
  {Kitaev}}, \bibinfo {author} {\bibfnamefont {A.}~\bibnamefont {Shen}}, \ and\
  \bibinfo {author} {\bibfnamefont {M.~N.}\ \bibnamefont {Vyalyi}},\ }\href
  {http://books.google.com.hk/books?id=qYHTvHPvmG8C} {\emph {\bibinfo {title}
  {{Classical and Quantum Computation}}}},\ Graduate studies in mathematics\
  (\bibinfo  {publisher} {American Mathematical Society},\ \bibinfo {year}
  {2002})\BibitemShut {NoStop}%
\bibitem [{\citenamefont {Buluta}\ and\ \citenamefont
  {Nori}(2009)}]{Buluta2009a}%
  \BibitemOpen
  \bibfield  {author} {\bibinfo {author} {\bibfnamefont {I.}~\bibnamefont
  {Buluta}}\ and\ \bibinfo {author} {\bibfnamefont {F.}~\bibnamefont {Nori}},\
  }\href {\doibase 10.1126/science.1177838} {\bibfield  {journal} {\bibinfo
  {journal} {Science}\ }\textbf {\bibinfo {volume} {326}},\ \bibinfo {pages}
  {108} (\bibinfo {year} {2009})}\BibitemShut {NoStop}%
\bibitem [{\citenamefont {Yung}\ \emph {et~al.}(2014)\citenamefont {Yung},
  \citenamefont {Whitfield}, \citenamefont {Boixo}, \citenamefont {Tempel},\
  and\ \citenamefont {Aspuru-Guzik}}]{Yung2012c}%
  \BibitemOpen
  \bibfield  {author} {\bibinfo {author} {\bibfnamefont {M.-h.}\ \bibnamefont
  {Yung}}, \bibinfo {author} {\bibfnamefont {J.~D.}\ \bibnamefont {Whitfield}},
  \bibinfo {author} {\bibfnamefont {S.}~\bibnamefont {Boixo}}, \bibinfo
  {author} {\bibfnamefont {D.~G.}\ \bibnamefont {Tempel}}, \ and\ \bibinfo
  {author} {\bibfnamefont {A.}~\bibnamefont {Aspuru-Guzik}},\ }in\ \href
  {\doibase 10.1002/9781118742631.ch03} {\emph {\bibinfo {booktitle} {Adv.
  Chem. Phys.}}},\ \bibinfo {series} {Advances in Chemical Physics}, Vol.\
  \bibinfo {volume} {154},\ \bibinfo {editor} {edited by\ \bibinfo {editor}
  {\bibfnamefont {S.}~\bibnamefont {Kais}}}\ (\bibinfo  {publisher} {John Wiley
  {\&} Sons, Inc.},\ \bibinfo {address} {Hoboken, New Jersey},\ \bibinfo {year}
  {2014})\ pp.\ \bibinfo {pages} {67--106},\ \Eprint
  {http://arxiv.org/abs/1203.1331} {arXiv:1203.1331} \BibitemShut {NoStop}%
\bibitem [{\citenamefont {Bennett}\ \emph {et~al.}(1993)\citenamefont
  {Bennett}, \citenamefont {Brassard}, \citenamefont {Cr{\'{e}}peau},
  \citenamefont {Jozsa}, \citenamefont {Peres},\ and\ \citenamefont
  {Wootters}}]{Bennett1998}%
  \BibitemOpen
  \bibfield  {author} {\bibinfo {author} {\bibfnamefont {C.~H.}\ \bibnamefont
  {Bennett}}, \bibinfo {author} {\bibfnamefont {G.}~\bibnamefont {Brassard}},
  \bibinfo {author} {\bibfnamefont {C.}~\bibnamefont {Cr{\'{e}}peau}}, \bibinfo
  {author} {\bibfnamefont {R.}~\bibnamefont {Jozsa}}, \bibinfo {author}
  {\bibfnamefont {A.}~\bibnamefont {Peres}}, \ and\ \bibinfo {author}
  {\bibfnamefont {W.~K.}\ \bibnamefont {Wootters}},\ }\href {\doibase
  10.1103/PhysRevLett.70.1895} {\bibfield  {journal} {\bibinfo  {journal}
  {Phys. Rev. Lett.}\ }\textbf {\bibinfo {volume} {70}},\ \bibinfo {pages}
  {1895} (\bibinfo {year} {1993})},\ \Eprint {http://arxiv.org/abs/9710013}
  {arXiv:9710013 [quant-ph]} \BibitemShut {NoStop}%
\bibitem [{\citenamefont {Giovannetti}\ \emph {et~al.}(2006)\citenamefont
  {Giovannetti}, \citenamefont {Lloyd},\ and\ \citenamefont
  {Maccone}}]{Giovannetti2006b}%
  \BibitemOpen
  \bibfield  {author} {\bibinfo {author} {\bibfnamefont {V.}~\bibnamefont
  {Giovannetti}}, \bibinfo {author} {\bibfnamefont {S.}~\bibnamefont {Lloyd}},
  \ and\ \bibinfo {author} {\bibfnamefont {L.}~\bibnamefont {Maccone}},\ }\href
  {\doibase 10.1103/PhysRevLett.96.010401} {\bibfield  {journal} {\bibinfo
  {journal} {Phys. Rev. Lett.}\ }\textbf {\bibinfo {volume} {96}},\ \bibinfo
  {pages} {010401} (\bibinfo {year} {2006})},\ \Eprint
  {http://arxiv.org/abs/arXiv:1102.2318v1} {arXiv:arXiv:1102.2318v1}
  \BibitemShut {NoStop}%
\bibitem [{\citenamefont {Lloyd}(2008)}]{Lloyd2008b}%
  \BibitemOpen
  \bibfield  {author} {\bibinfo {author} {\bibfnamefont {S.}~\bibnamefont
  {Lloyd}},\ }\href {\doibase 10.1126/science.1160627} {\bibfield  {journal}
  {\bibinfo  {journal} {Science}\ }\textbf {\bibinfo {volume} {321}},\ \bibinfo
  {pages} {1463} (\bibinfo {year} {2008})}\BibitemShut {NoStop}%
\bibitem [{\citenamefont {Tan}\ \emph {et~al.}(2008)\citenamefont {Tan},
  \citenamefont {Erkmen}, \citenamefont {Giovannetti}, \citenamefont {Guha},
  \citenamefont {Lloyd}, \citenamefont {Maccone}, \citenamefont {Pirandola},\
  and\ \citenamefont {Shapiro}}]{Tan2008}%
  \BibitemOpen
  \bibfield  {author} {\bibinfo {author} {\bibfnamefont {S.-H.}\ \bibnamefont
  {Tan}}, \bibinfo {author} {\bibfnamefont {B.~I.}\ \bibnamefont {Erkmen}},
  \bibinfo {author} {\bibfnamefont {V.}~\bibnamefont {Giovannetti}}, \bibinfo
  {author} {\bibfnamefont {S.}~\bibnamefont {Guha}}, \bibinfo {author}
  {\bibfnamefont {S.}~\bibnamefont {Lloyd}}, \bibinfo {author} {\bibfnamefont
  {L.}~\bibnamefont {Maccone}}, \bibinfo {author} {\bibfnamefont
  {S.}~\bibnamefont {Pirandola}}, \ and\ \bibinfo {author} {\bibfnamefont
  {J.~H.}\ \bibnamefont {Shapiro}},\ }\href {\doibase
  10.1103/PhysRevLett.101.253601} {\bibfield  {journal} {\bibinfo  {journal}
  {Phys. Rev. Lett.}\ }\textbf {\bibinfo {volume} {101}},\ \bibinfo {pages}
  {253601} (\bibinfo {year} {2008})},\ \Eprint {http://arxiv.org/abs/0810.0534}
  {arXiv:0810.0534} \BibitemShut {NoStop}%
\bibitem [{\citenamefont {Shapiro}(2009)}]{Shapiro2009}%
  \BibitemOpen
  \bibfield  {author} {\bibinfo {author} {\bibfnamefont {J.~H.}\ \bibnamefont
  {Shapiro}},\ }\href {\doibase 10.1103/PhysRevA.80.022320} {\bibfield
  {journal} {\bibinfo  {journal} {Phys. Rev. A}\ }\textbf {\bibinfo {volume}
  {80}},\ \bibinfo {pages} {22320} (\bibinfo {year} {2009})}\BibitemShut
  {NoStop}%
\bibitem [{\citenamefont {Guha}(2009)}]{Guha2009}%
  \BibitemOpen
  \bibfield  {author} {\bibinfo {author} {\bibfnamefont {S.}~\bibnamefont
  {Guha}},\ }in\ \href {\doibase 10.1109/ISIT.2009.5205594} {\emph {\bibinfo
  {booktitle} {2009 IEEE Int. Symp. Inf. Theory}}}\ (\bibinfo  {publisher}
  {IEEE},\ \bibinfo {year} {2009})\ pp.\ \bibinfo {pages} {963--967},\ \Eprint
  {http://arxiv.org/abs/0902.2932} {arXiv:0902.2932} \BibitemShut {NoStop}%
\bibitem [{\citenamefont {Lopaeva}\ \emph {et~al.}(2013)\citenamefont
  {Lopaeva}, \citenamefont {{Ruo Berchera}}, \citenamefont {Degiovanni},
  \citenamefont {Olivares}, \citenamefont {Brida},\ and\ \citenamefont
  {Genovese}}]{Lopaeva2013}%
  \BibitemOpen
  \bibfield  {author} {\bibinfo {author} {\bibfnamefont {E.~D.}\ \bibnamefont
  {Lopaeva}}, \bibinfo {author} {\bibfnamefont {I.}~\bibnamefont {{Ruo
  Berchera}}}, \bibinfo {author} {\bibfnamefont {I.~P.}\ \bibnamefont
  {Degiovanni}}, \bibinfo {author} {\bibfnamefont {S.}~\bibnamefont
  {Olivares}}, \bibinfo {author} {\bibfnamefont {G.}~\bibnamefont {Brida}}, \
  and\ \bibinfo {author} {\bibfnamefont {M.}~\bibnamefont {Genovese}},\ }\href
  {\doibase 10.1103/PhysRevLett.110.153603} {\bibfield  {journal} {\bibinfo
  {journal} {Phys. Rev. Lett.}\ }\textbf {\bibinfo {volume} {110}},\ \bibinfo
  {pages} {153603} (\bibinfo {year} {2013})},\ \Eprint
  {http://arxiv.org/abs/arXiv:1303.4304v1} {arXiv:arXiv:1303.4304v1}
  \BibitemShut {NoStop}%
\bibitem [{\citenamefont {Barzanjeh}\ \emph {et~al.}(2015)\citenamefont
  {Barzanjeh}, \citenamefont {Guha}, \citenamefont {Weedbrook}, \citenamefont
  {Vitali}, \citenamefont {Shapiro},\ and\ \citenamefont
  {Pirandola}}]{Barzanjeh2014}%
  \BibitemOpen
  \bibfield  {author} {\bibinfo {author} {\bibfnamefont {S.}~\bibnamefont
  {Barzanjeh}}, \bibinfo {author} {\bibfnamefont {S.}~\bibnamefont {Guha}},
  \bibinfo {author} {\bibfnamefont {C.}~\bibnamefont {Weedbrook}}, \bibinfo
  {author} {\bibfnamefont {D.}~\bibnamefont {Vitali}}, \bibinfo {author}
  {\bibfnamefont {J.~H.}\ \bibnamefont {Shapiro}}, \ and\ \bibinfo {author}
  {\bibfnamefont {S.}~\bibnamefont {Pirandola}},\ }\href {\doibase
  10.1103/PhysRevLett.114.080503} {\bibfield  {journal} {\bibinfo  {journal}
  {Phys. Rev. Lett.}\ }\textbf {\bibinfo {volume} {114}},\ \bibinfo {pages}
  {80503} (\bibinfo {year} {2015})},\ \Eprint {http://arxiv.org/abs/1410.4008}
  {arXiv:1410.4008} \BibitemShut {NoStop}%
\bibitem [{\citenamefont {Zhang}\ \emph
  {et~al.}(2014{\natexlab{a}})\citenamefont {Zhang}, \citenamefont {Zou},
  \citenamefont {Shi}, \citenamefont {Guo},\ and\ \citenamefont
  {Guo}}]{Zhang2014}%
  \BibitemOpen
  \bibfield  {author} {\bibinfo {author} {\bibfnamefont {S.}~\bibnamefont
  {Zhang}}, \bibinfo {author} {\bibfnamefont {X.}~\bibnamefont {Zou}}, \bibinfo
  {author} {\bibfnamefont {J.}~\bibnamefont {Shi}}, \bibinfo {author}
  {\bibfnamefont {J.}~\bibnamefont {Guo}}, \ and\ \bibinfo {author}
  {\bibfnamefont {G.}~\bibnamefont {Guo}},\ }\href {\doibase
  10.1103/PhysRevA.90.052308} {\bibfield  {journal} {\bibinfo  {journal} {Phys.
  Rev. A}\ }\textbf {\bibinfo {volume} {90}},\ \bibinfo {pages} {52308}
  (\bibinfo {year} {2014}{\natexlab{a}})}\BibitemShut {NoStop}%
\bibitem [{\citenamefont {Zhang}\ \emph
  {et~al.}(2014{\natexlab{b}})\citenamefont {Zhang}, \citenamefont {Guo},
  \citenamefont {Bao}, \citenamefont {Shi}, \citenamefont {Jin}, \citenamefont
  {Zou},\ and\ \citenamefont {Guo}}]{Zhang2014a}%
  \BibitemOpen
  \bibfield  {author} {\bibinfo {author} {\bibfnamefont {S.}~\bibnamefont
  {Zhang}}, \bibinfo {author} {\bibfnamefont {J.}~\bibnamefont {Guo}}, \bibinfo
  {author} {\bibfnamefont {W.}~\bibnamefont {Bao}}, \bibinfo {author}
  {\bibfnamefont {J.}~\bibnamefont {Shi}}, \bibinfo {author} {\bibfnamefont
  {C.}~\bibnamefont {Jin}}, \bibinfo {author} {\bibfnamefont {X.}~\bibnamefont
  {Zou}}, \ and\ \bibinfo {author} {\bibfnamefont {G.}~\bibnamefont {Guo}},\
  }\href {\doibase 10.1103/PhysRevA.89.062309} {\bibfield  {journal} {\bibinfo
  {journal} {Phys. Rev. A}\ }\textbf {\bibinfo {volume} {89}},\ \bibinfo
  {pages} {62309} (\bibinfo {year} {2014}{\natexlab{b}})}\BibitemShut {NoStop}%
\bibitem [{\citenamefont {Bradshaw}\ \emph {et~al.}(2017)\citenamefont
  {Bradshaw}, \citenamefont {Assad}, \citenamefont {Haw}, \citenamefont {Tan},
  \citenamefont {Lam},\ and\ \citenamefont {Gu}}]{Bradshaw2016}%
  \BibitemOpen
  \bibfield  {author} {\bibinfo {author} {\bibfnamefont {M.}~\bibnamefont
  {Bradshaw}}, \bibinfo {author} {\bibfnamefont {S.~M.}\ \bibnamefont {Assad}},
  \bibinfo {author} {\bibfnamefont {J.~Y.}\ \bibnamefont {Haw}}, \bibinfo
  {author} {\bibfnamefont {S.-h.}\ \bibnamefont {Tan}}, \bibinfo {author}
  {\bibfnamefont {P.~K.}\ \bibnamefont {Lam}}, \ and\ \bibinfo {author}
  {\bibfnamefont {M.}~\bibnamefont {Gu}},\ }\href {\doibase
  10.1103/PhysRevA.95.022333} {\bibfield  {journal} {\bibinfo  {journal} {Phys.
  Rev. A}\ }\textbf {\bibinfo {volume} {95}},\ \bibinfo {pages} {022333}
  (\bibinfo {year} {2017})},\ \Eprint {http://arxiv.org/abs/1611.10020}
  {arXiv:1611.10020} \BibitemShut {NoStop}%
\bibitem [{\citenamefont {Sanz}\ \emph {et~al.}(2017)\citenamefont {Sanz},
  \citenamefont {{Las Heras}}, \citenamefont {Garc{\'{i}}a-Ripoll},
  \citenamefont {Solano},\ and\ \citenamefont {{Di Candia}}}]{Sanz2016}%
  \BibitemOpen
  \bibfield  {author} {\bibinfo {author} {\bibfnamefont {M.}~\bibnamefont
  {Sanz}}, \bibinfo {author} {\bibfnamefont {U.}~\bibnamefont {{Las Heras}}},
  \bibinfo {author} {\bibfnamefont {J.~J.}\ \bibnamefont
  {Garc{\'{i}}a-Ripoll}}, \bibinfo {author} {\bibfnamefont {E.}~\bibnamefont
  {Solano}}, \ and\ \bibinfo {author} {\bibfnamefont {R.}~\bibnamefont {{Di
  Candia}}},\ }\href {\doibase 10.1103/PhysRevLett.118.070803} {\bibfield
  {journal} {\bibinfo  {journal} {Phys. Rev. Lett.}\ }\textbf {\bibinfo
  {volume} {118}},\ \bibinfo {pages} {070803} (\bibinfo {year} {2017})},\
  \Eprint {http://arxiv.org/abs/1606.06656} {arXiv:1606.06656} \BibitemShut
  {NoStop}%
\bibitem [{\citenamefont {Lanzagorta}\ \emph {et~al.}(2016)\citenamefont
  {Lanzagorta}, \citenamefont {Uhlmann}, \citenamefont {Le}, \citenamefont
  {Jitrik},\ and\ \citenamefont {Venegas-Andraca}}]{Lanzagorta2016}%
  \BibitemOpen
  \bibfield  {author} {\bibinfo {author} {\bibfnamefont {M.}~\bibnamefont
  {Lanzagorta}}, \bibinfo {author} {\bibfnamefont {J.}~\bibnamefont {Uhlmann}},
  \bibinfo {author} {\bibfnamefont {T.}~\bibnamefont {Le}}, \bibinfo {author}
  {\bibfnamefont {O.}~\bibnamefont {Jitrik}}, \ and\ \bibinfo {author}
  {\bibfnamefont {S.~E.}\ \bibnamefont {Venegas-Andraca}}\ }(\bibinfo {year}
  {2016})\ p.\ \bibinfo {pages} {98291D}\BibitemShut {NoStop}%
\bibitem [{\citenamefont {Liu}\ \emph {et~al.}(2017)\citenamefont {Liu},
  \citenamefont {Zhang}, \citenamefont {Gu},\ and\ \citenamefont
  {Li}}]{Liu2017}%
  \BibitemOpen
  \bibfield  {author} {\bibinfo {author} {\bibfnamefont {K.}~\bibnamefont
  {Liu}}, \bibinfo {author} {\bibfnamefont {Q.-W.}\ \bibnamefont {Zhang}},
  \bibinfo {author} {\bibfnamefont {Y.-J.}\ \bibnamefont {Gu}}, \ and\ \bibinfo
  {author} {\bibfnamefont {Q.-L.}\ \bibnamefont {Li}},\ }\href {\doibase
  10.1103/PhysRevA.95.042317} {\bibfield  {journal} {\bibinfo  {journal} {Phys.
  Rev. A}\ }\textbf {\bibinfo {volume} {95}},\ \bibinfo {pages} {042317}
  (\bibinfo {year} {2017})}\BibitemShut {NoStop}%
\bibitem [{\citenamefont {Zhuang}\ \emph
  {et~al.}(2017{\natexlab{a}})\citenamefont {Zhuang}, \citenamefont {Zhang},\
  and\ \citenamefont {Shapiro}}]{Zhuang2017}%
  \BibitemOpen
  \bibfield  {author} {\bibinfo {author} {\bibfnamefont {Q.}~\bibnamefont
  {Zhuang}}, \bibinfo {author} {\bibfnamefont {Z.}~\bibnamefont {Zhang}}, \
  and\ \bibinfo {author} {\bibfnamefont {J.~H.}\ \bibnamefont {Shapiro}},\
  }\href {\doibase 10.1103/PhysRevLett.118.040801} {\bibfield  {journal}
  {\bibinfo  {journal} {Phys. Rev. Lett.}\ }\textbf {\bibinfo {volume} {118}},\
  \bibinfo {pages} {40801} (\bibinfo {year} {2017}{\natexlab{a}})}\BibitemShut
  {NoStop}%
\bibitem [{\citenamefont {Zhuang}\ \emph
  {et~al.}(2017{\natexlab{b}})\citenamefont {Zhuang}, \citenamefont {Zhang},\
  and\ \citenamefont {Shapiro}}]{Zhuang2017a}%
  \BibitemOpen
  \bibfield  {author} {\bibinfo {author} {\bibfnamefont {Q.}~\bibnamefont
  {Zhuang}}, \bibinfo {author} {\bibfnamefont {Z.}~\bibnamefont {Zhang}}, \
  and\ \bibinfo {author} {\bibfnamefont {J.~H.}\ \bibnamefont {Shapiro}},\
  }\href {\doibase 10.1364/JOSAB.34.001567} {\bibfield  {journal} {\bibinfo
  {journal} {J. Opt. Soc. Am. B}\ }\textbf {\bibinfo {volume} {34}},\ \bibinfo
  {pages} {1567} (\bibinfo {year} {2017}{\natexlab{b}})},\ \Eprint
  {http://arxiv.org/abs/1703.02463} {arXiv:1703.02463} \BibitemShut {NoStop}%
\bibitem [{\citenamefont {{Las Heras}}\ \emph {et~al.}(2017)\citenamefont {{Las
  Heras}}, \citenamefont {{Di Candia}}, \citenamefont {Fedorov}, \citenamefont
  {Deppe}, \citenamefont {Sanz},\ and\ \citenamefont {Solano}}]{LasHeras2017}%
  \BibitemOpen
  \bibfield  {author} {\bibinfo {author} {\bibfnamefont {U.}~\bibnamefont {{Las
  Heras}}}, \bibinfo {author} {\bibfnamefont {R.}~\bibnamefont {{Di Candia}}},
  \bibinfo {author} {\bibfnamefont {K.~G.}\ \bibnamefont {Fedorov}}, \bibinfo
  {author} {\bibfnamefont {F.}~\bibnamefont {Deppe}}, \bibinfo {author}
  {\bibfnamefont {M.}~\bibnamefont {Sanz}}, \ and\ \bibinfo {author}
  {\bibfnamefont {E.}~\bibnamefont {Solano}},\ }\href {\doibase
  10.1038/s41598-017-08505-w} {\bibfield  {journal} {\bibinfo  {journal} {Sci.
  Rep.}\ }\textbf {\bibinfo {volume} {7}},\ \bibinfo {pages} {9333} (\bibinfo
  {year} {2017})},\ \Eprint {http://arxiv.org/abs/1611.10280}
  {arXiv:1611.10280} \BibitemShut {NoStop}%
\bibitem [{\citenamefont {Ruskai}(2003)}]{Ruskai2003}%
  \BibitemOpen
  \bibfield  {author} {\bibinfo {author} {\bibfnamefont {M.~B.}\ \bibnamefont
  {Ruskai}},\ }\href {\doibase 10.1142/S0129055X03001710} {\bibfield  {journal}
  {\bibinfo  {journal} {Rev. Math. Phys.}\ }\textbf {\bibinfo {volume} {15}},\
  \bibinfo {pages} {643} (\bibinfo {year} {2003})},\ \Eprint
  {http://arxiv.org/abs/0302032} {arXiv:0302032 [quant-ph]} \BibitemShut
  {NoStop}%
\bibitem [{\citenamefont {Xu}\ \emph {et~al.}(2011)\citenamefont {Xu},
  \citenamefont {Shapiro}, \citenamefont {Ralph},\ and\ \citenamefont
  {Lam}}]{Xu2011b}%
  \BibitemOpen
  \bibfield  {author} {\bibinfo {author} {\bibfnamefont {W.}~\bibnamefont
  {Xu}}, \bibinfo {author} {\bibfnamefont {J.~H.}\ \bibnamefont {Shapiro}},
  \bibinfo {author} {\bibfnamefont {T.}~\bibnamefont {Ralph}}, \ and\ \bibinfo
  {author} {\bibfnamefont {P.~K.}\ \bibnamefont {Lam}},\ }in\ \href {\doibase
  10.1063/1.3630142} {\emph {\bibinfo {booktitle} {AIP Conf. Proc.}}},\ Vol.\
  \bibinfo {volume} {1363}\ (\bibinfo {year} {2011})\ pp.\ \bibinfo {pages}
  {31--34},\ \Eprint {http://arxiv.org/abs/0904.2490} {arXiv:0904.2490}
  \BibitemShut {NoStop}%
\bibitem [{\citenamefont {Zhang}\ \emph {et~al.}(2013)\citenamefont {Zhang},
  \citenamefont {Tengner}, \citenamefont {Zhong}, \citenamefont {Wong},\ and\
  \citenamefont {Shapiro}}]{Zhang2013a}%
  \BibitemOpen
  \bibfield  {author} {\bibinfo {author} {\bibfnamefont {Z.}~\bibnamefont
  {Zhang}}, \bibinfo {author} {\bibfnamefont {M.}~\bibnamefont {Tengner}},
  \bibinfo {author} {\bibfnamefont {T.}~\bibnamefont {Zhong}}, \bibinfo
  {author} {\bibfnamefont {F.~N.~C.}\ \bibnamefont {Wong}}, \ and\ \bibinfo
  {author} {\bibfnamefont {J.~H.}\ \bibnamefont {Shapiro}},\ }\href {\doibase
  10.1103/PhysRevLett.111.010501} {\bibfield  {journal} {\bibinfo  {journal}
  {Phys. Rev. Lett.}\ }\textbf {\bibinfo {volume} {111}},\ \bibinfo {pages}
  {10501} (\bibinfo {year} {2013})}\BibitemShut {NoStop}%
\bibitem [{\citenamefont {Ralph}\ and\ \citenamefont {Lam}(2013)}]{Ralph2013a}%
  \BibitemOpen
  \bibfield  {author} {\bibinfo {author} {\bibfnamefont {T.}~\bibnamefont
  {Ralph}}\ and\ \bibinfo {author} {\bibfnamefont {P.}~\bibnamefont {Lam}},\
  }\href {\doibase 10.1103/Physics.6.74} {\bibfield  {journal} {\bibinfo
  {journal} {Physics (College. Park. Md).}\ }\textbf {\bibinfo {volume} {6}},\
  \bibinfo {pages} {74} (\bibinfo {year} {2013})}\BibitemShut {NoStop}%
\bibitem [{\citenamefont {Shapiro}\ \emph {et~al.}(2014)\citenamefont
  {Shapiro}, \citenamefont {Zhang},\ and\ \citenamefont {Wong}}]{Shapiro2014}%
  \BibitemOpen
  \bibfield  {author} {\bibinfo {author} {\bibfnamefont {J.~H.}\ \bibnamefont
  {Shapiro}}, \bibinfo {author} {\bibfnamefont {Z.}~\bibnamefont {Zhang}}, \
  and\ \bibinfo {author} {\bibfnamefont {F.~N.~C.}\ \bibnamefont {Wong}},\
  }\href {\doibase 10.1007/s11128-013-0662-1} {\bibfield  {journal} {\bibinfo
  {journal} {Quantum Inf. Process.}\ }\textbf {\bibinfo {volume} {13}},\
  \bibinfo {pages} {2171} (\bibinfo {year} {2014})}\BibitemShut {NoStop}%
\bibitem [{\citenamefont {Zhang}\ \emph {et~al.}(2015)\citenamefont {Zhang},
  \citenamefont {Mouradian}, \citenamefont {Wong},\ and\ \citenamefont
  {Shapiro}}]{Zhang2015}%
  \BibitemOpen
  \bibfield  {author} {\bibinfo {author} {\bibfnamefont {Z.}~\bibnamefont
  {Zhang}}, \bibinfo {author} {\bibfnamefont {S.}~\bibnamefont {Mouradian}},
  \bibinfo {author} {\bibfnamefont {F.~N.~C.}\ \bibnamefont {Wong}}, \ and\
  \bibinfo {author} {\bibfnamefont {J.~H.}\ \bibnamefont {Shapiro}},\ }\href
  {\doibase 10.1103/PhysRevLett.114.110506} {\bibfield  {journal} {\bibinfo
  {journal} {Phys. Rev. Lett.}\ }\textbf {\bibinfo {volume} {114}},\ \bibinfo
  {pages} {110506} (\bibinfo {year} {2015})}\BibitemShut {NoStop}%
\bibitem [{\citenamefont {Harrow}\ \emph {et~al.}(2010)\citenamefont {Harrow},
  \citenamefont {Hassidim}, \citenamefont {Leung},\ and\ \citenamefont
  {Watrous}}]{Harrow2010}%
  \BibitemOpen
  \bibfield  {author} {\bibinfo {author} {\bibfnamefont {A.~W.}\ \bibnamefont
  {Harrow}}, \bibinfo {author} {\bibfnamefont {A.}~\bibnamefont {Hassidim}},
  \bibinfo {author} {\bibfnamefont {D.~W.}\ \bibnamefont {Leung}}, \ and\
  \bibinfo {author} {\bibfnamefont {J.}~\bibnamefont {Watrous}},\ }\href
  {\doibase 10.1103/PhysRevA.81.032339} {\bibfield  {journal} {\bibinfo
  {journal} {Phys. Rev. A}\ }\textbf {\bibinfo {volume} {81}},\ \bibinfo
  {pages} {032339} (\bibinfo {year} {2010})},\ \Eprint
  {http://arxiv.org/abs/0909.0256} {arXiv:0909.0256} \BibitemShut {NoStop}%
\bibitem [{\citenamefont {Rosgen}\ and\ \citenamefont
  {Watrous}(2005)}]{Rosgen2005}%
  \BibitemOpen
  \bibfield  {author} {\bibinfo {author} {\bibfnamefont {B.}~\bibnamefont
  {Rosgen}}\ and\ \bibinfo {author} {\bibfnamefont {J.}~\bibnamefont
  {Watrous}},\ }in\ \href {\doibase 10.1109/CCC.2005.21} {\emph {\bibinfo
  {booktitle} {20th Annu. IEEE Conf. Comput. Complex.}}}\ (\bibinfo
  {publisher} {IEEE},\ \bibinfo {year} {2005})\ pp.\ \bibinfo {pages}
  {344--354},\ \Eprint {http://arxiv.org/abs/0407056} {arXiv:0407056 [cs]}
  \BibitemShut {NoStop}%
\bibitem [{\citenamefont {Jain}\ \emph {et~al.}(2010)\citenamefont {Jain},
  \citenamefont {Ji}, \citenamefont {Upadhyay},\ and\ \citenamefont
  {Watrous}}]{Jain2010}%
  \BibitemOpen
  \bibfield  {author} {\bibinfo {author} {\bibfnamefont {R.}~\bibnamefont
  {Jain}}, \bibinfo {author} {\bibfnamefont {Z.}~\bibnamefont {Ji}}, \bibinfo
  {author} {\bibfnamefont {S.}~\bibnamefont {Upadhyay}}, \ and\ \bibinfo
  {author} {\bibfnamefont {J.}~\bibnamefont {Watrous}},\ }\href {\doibase
  10.1145/1859204.1859231} {\bibfield  {journal} {\bibinfo  {journal} {Commun.
  ACM}\ }\textbf {\bibinfo {volume} {53}},\ \bibinfo {pages} {102} (\bibinfo
  {year} {2010})}\BibitemShut {NoStop}%
\bibitem [{\citenamefont {Modi}\ \emph {et~al.}(2012)\citenamefont {Modi},
  \citenamefont {Brodutch}, \citenamefont {Cable}, \citenamefont {Paterek},\
  and\ \citenamefont {Vedral}}]{Modi2012a}%
  \BibitemOpen
  \bibfield  {author} {\bibinfo {author} {\bibfnamefont {K.}~\bibnamefont
  {Modi}}, \bibinfo {author} {\bibfnamefont {A.}~\bibnamefont {Brodutch}},
  \bibinfo {author} {\bibfnamefont {H.}~\bibnamefont {Cable}}, \bibinfo
  {author} {\bibfnamefont {T.}~\bibnamefont {Paterek}}, \ and\ \bibinfo
  {author} {\bibfnamefont {V.}~\bibnamefont {Vedral}},\ }\href {\doibase
  10.1103/RevModPhys.84.1655} {\bibfield  {journal} {\bibinfo  {journal} {Rev.
  Mod. Phys.}\ }\textbf {\bibinfo {volume} {84}},\ \bibinfo {pages} {1655}
  (\bibinfo {year} {2012})}\BibitemShut {NoStop}%
\bibitem [{\citenamefont {Weedbrook}\ \emph {et~al.}(2013)\citenamefont
  {Weedbrook}, \citenamefont {Pirandola}, \citenamefont {Thompson},
  \citenamefont {Vedral},\ and\ \citenamefont {Gu}}]{Weedbrook2013}%
  \BibitemOpen
  \bibfield  {author} {\bibinfo {author} {\bibfnamefont {C.}~\bibnamefont
  {Weedbrook}}, \bibinfo {author} {\bibfnamefont {S.}~\bibnamefont
  {Pirandola}}, \bibinfo {author} {\bibfnamefont {J.}~\bibnamefont {Thompson}},
  \bibinfo {author} {\bibfnamefont {V.}~\bibnamefont {Vedral}}, \ and\ \bibinfo
  {author} {\bibfnamefont {M.}~\bibnamefont {Gu}},\ }\href
  {http://arxiv.org/abs/1312.3332} {\ ,\ \bibinfo {pages} {7} (\bibinfo {year}
  {2013})},\ \Eprint {http://arxiv.org/abs/1312.3332} {arXiv:1312.3332}
  \BibitemShut {NoStop}%
\bibitem [{\citenamefont {Helstrom}(1969)}]{Helstrom1969}%
  \BibitemOpen
  \bibfield  {author} {\bibinfo {author} {\bibfnamefont {C.~W.}\ \bibnamefont
  {Helstrom}},\ }\href {\doibase 10.1007/BF01007479} {\bibfield  {journal}
  {\bibinfo  {journal} {J. Stat. Phys.}\ }\textbf {\bibinfo {volume} {1}},\
  \bibinfo {pages} {231} (\bibinfo {year} {1969})}\BibitemShut {NoStop}%
\end{thebibliography}%

%\onecolumn

\section*{Appendix: Proofs of the theorems}

\section*{Results for region I $\&$ II}

%\begin{theorem*}[\bf Region I for both classical and quantum illuminations]\label{theo:con_pp0l12}
{\bf Result: Region I for both classical and quantum illuminations}
Suppose 
\begin{itemize}
\item[(i)] $p_0 \le 1/2$, and
\item[(ii)] $\eta \le \eta_* \equiv 1 - p_0/p_1$ (or equivalently, $\gamma \ge 0$),
\end{itemize}
then (a) the minimal errors for conventional illumination and quantum illumination are equal to $p_0$, i.e.,
\begin{equation}\label{Pcerr_p0}
  {P}_{{\rm{err}}} = p_0 \ ,
\end{equation}
and (b) the bound can be achieved with any (pure or mixed) state.
%\end{theorem*}

\begin{proof}
If $p_0 \le 1/2$ and $\eta \le \eta_*$, we have $\gamma  \geqslant 0$, which also implies that the Hermitian matrix, ${\Omega _{\text{c(q)}}}$, is a positive sum of two density matrices (with positive eigenvalues). 

Consequently, all eigenvalues ${\lambda _i} = \left\langle i \right|{\Omega _{\text{c(q)}}}\left| i \right\rangle $, with an eigenvector~$\left| i \right\rangle$, are positive, i.e., ${\lambda _i} \geqslant 0$. In this case the trace norm of ${{\Omega _{\text{c(q)}}}}$ can be obtained directly by taking the trace, i.e., 
\begin{equation}
\left\| {{\Omega _{\text{c(q)}}}} \right\| = \left| {{p_1}\eta  + \gamma } \right| = {p_1} - {p_0} \ ,
\end{equation}
which implies the result stated in Eq.~(\ref{Pcerr_p0}). Note that the whole argument is applicable to any pure state.
\end{proof}

%\begin{theorem*}[\bf Region II for conventional illumination]\label{theo_con_p0g12}
{\bf Result: Region II for conventional illumination}
Suppose 
\begin{itemize}
\item[(i)] ${p_0} \geqslant 1/2$, and
\item[(ii)]  $\eta  \leqslant -\eta_*( {\tfrac{{{\lambda _{\min }}}}{{1 - {\lambda _{\min }}}}} )$ ,
\end{itemize}
with ${\lambda _{\min }} = {\lambda _{\min }}\left( {{\rho _E}} \right) \geqslant 0$, the minimal eigenvalue of the environment signal $\rho_E$, then (a) the minimal error for conventional illumination is equal to $p_1$, i.e.,
 \begin{equation}\label{Pcerr_p1}
  {P}_{{\rm{err}}}^{{c}} = p_1 \ ,
\end{equation}
and (b) the bound can be achieved with any (pure or mixed) state.
%\end{theorem*}

\begin{proof}
Let's consider the conventional illumination first. We now express $\Omega_{\rm c}$ as follows： 
\begin{equation}
  {\Omega _{\text{c}}} = \left( {{p_1}\eta  + \gamma {\lambda _{\min }}} \right)\rho + \gamma \,\left( {{\rho _E} - {\lambda _{\min }}\rho} \right) \ .
\end{equation}
Note that the matrix, ${{\rho _E} - {\lambda _{\min }}\rho}$, contains non-negative eigenvalues. Now suppose the following conditions are satisfied, (i) ${p_0} \geqslant 1/2$ (or equivalently ${p_0} \geqslant {p_1}$), which implies that $\gamma  \equiv {p_1}\left( {1 - \eta } \right) - {p_0} \leqslant {p_0}\left( {1 - \eta } \right) - {p_0} =  - {p_0}\eta  \leqslant 0$, and (ii) $\eta  \leqslant ( {\tfrac{{{p_0}}}{{{p_1}}} - 1} )( {\tfrac{{{\lambda _{\min }}}}{{1 - {\lambda _{\min }}}}} )$, which further implies that 
\begin{equation}
{p_1}\eta  + \gamma {\lambda _{\min }} \leqslant 0 \ ,
\end{equation}
then the trace norm of ${\Omega _c}$ is given by the trace of ${\Omega _c}$, i.e., 
\begin{equation}
\left\| {{\Omega _c}\left( {\rho} \right)} \right\| = \left| {{p_1}\eta  + \gamma } \right| = {p_0} - {p_1}\ ,
\end{equation}
which implies the result in Eq.~(\ref{Pcerr_p1}). The proof for quantum case is similar.
\end{proof}

%%%%%%%%%%%%%%%%%%%%%%%%%%%%
\section*{Result for region III}

For conventional illumination in region III, we have $\gamma < 0$. Therefore, we can always write 
\begin{equation}
  {\Omega _c}\left( {\left| \psi  \right\rangle \left\langle \psi  \right|} \right) \propto {\rho _E} - \alpha \left| \psi  \right\rangle \left\langle \psi  \right| \ ,
\end{equation}
for some $\alpha > 0$. We shall see that (i) the trace norm of $\Omega_c$ is determined by the minimum eigenvalue of the matrix
\begin{equation}
{\rho _E} - \alpha \left| \psi  \right\rangle \left\langle \psi  \right| \ ,
\end{equation}
and (ii) the smallest eigenvalue is minimized by choosing the signal state as the eigenstate $\left| {{e_k}} \right\rangle $ with the minimal eigenvalue. These results come from the following lemmas.

\begin{lemma}[\bf Positivity of eigenvalues I]
Suppose $\left\langle \psi  \right|\left. {{e_k}} \right\rangle  \ne 0$ for some $k$'s, the eigenvalues, namely ${E_1} \geqslant {E_2} \geqslant ... \geqslant {E_d}$, of a $d$-dimensional matrix of the form, $\rho  - \alpha \left| \psi  \right\rangle \left\langle \psi  \right|$, where $\alpha>0$ and $\left\langle \psi  \right|\rho \left| \psi  \right\rangle  \ne 0$, can have at most one positive eigenvalue, i.e., ${E_1} \geqslant {E_2} \geqslant ... \geqslant {E_{d - 1}} \geqslant 0$ but the smallest eigenvalue ${E_d}$ may be positive, negative, or zero.
\end{lemma}

\begin{proof}
Consider the eigenvalue equation for the same matrix, i.e., $\left( {\rho  - \alpha \left| \psi  \right\rangle \left\langle \psi  \right|} \right)\left| {{e_k}} \right\rangle  = {E_k}\left| {{e_k}} \right\rangle $, for any $k \in \{1,2,3,..,d \}$, which can be written as 
\begin{equation}
\left( {\rho  - {E_k}I} \right)| {{e_k}} \rangle  = \alpha \left| \psi  \right\rangle \left\langle \psi  \right| {{e_k}} \rangle
\end{equation}
or $\left| {{e_k}} \right\rangle  = \alpha {\left( {\rho  - {E_k}I} \right)^{ - 1}}\left| \psi  \right\rangle \left\langle \psi  \right|\left. {{e_k}} \right\rangle $. Now, as $\left\langle \psi  \right|\left. {{e_k}} \right\rangle  \ne 0$, we therefore have the following relation: 
\begin{equation}
1 = \alpha \left\langle \psi  \right|{\left( {\rho  - {E_k}I} \right)^{ - 1}}\left| \psi  \right\rangle \ .
\end{equation}
Therefore, in terms of the eigenvalues $\lambda_k$ and eigenvectors ${\theta_k}$ of $\rho  = \sum\nolimits_k {{\lambda _k}} \left| {{\theta _k}} \right\rangle \left\langle {{\theta _k}} \right|$, the eigenvalues $E_k$ are the roots of the equation.
\begin{equation}
\frac{1}{\alpha } = \sum\limits_k {\frac{{|\left\langle {{\theta _k}} \right.\left| \psi  \right\rangle {|^2}}}{{{\lambda_k} - E}}} \ .
\end{equation}
Now, the right-hand side increases monotonically from zero as $E$ increases from $ - \infty $ to zero. Therefore, depending on the value of $\alpha$, there can be, at most, one negative eigenvalue for the matrix $\rho  - \alpha \left| \psi  \right\rangle \left\langle \psi  \right|$.
\end{proof}

The same result can be derived in an alternatively way, as follows.

\begin{lemma}[\bf Positivity of eigenvalues II]
	Let $\rho$ be a $d$-dimentional density matrix, let $| \psi \rangle$ be a
	pure quantum state, and $\alpha > 0$ a real number. Then there is at most one negative eigenvalue of the operator $\rho - \alpha
	\left| \psi \rangle \langle \psi \right|$. 
\end{lemma}

\begin{proof}	
	Suppose we can find two distinct negative eigenvalues such that $E_d < E_{d - 1} < 0$. Consider
	the subspace spanned by the corresponding eigenvectors, $V = Span \left\{ |
	e_d \rangle, | e_{d - 1} \rangle \right\}$, of the matrix $\rho - \alpha
	\left| \psi \rangle \langle \psi \right|$.  Clearly, there exists a linear 	combination, denoted by
\begin{equation}
\left| {{\psi ^ \bot }} \right\rangle  \equiv \beta  | e_d \rangle + \theta  | e_{d -
		1} \rangle \neq 0 \ ,
\end{equation}	
which is orthogonal to $| \psi \rangle$, i.e., $\left\langle \psi  \right.\left| {{\psi ^ \bot }} \right\rangle  = 0$. Then, $\rho -
	\alpha  | \psi \rangle \langle \psi |$ maps $| {\psi ^ \bot } \rangle$ to $\rho | {\psi ^ \bot }
	\rangle$, i.e.,
\begin{equation}
\left( {\rho  - \alpha |\psi \rangle \langle \psi |} \right)\left| {{\psi ^ \bot }} \right\rangle  = \rho \left| {{\psi ^ \bot }} \right\rangle \ .
\end{equation}	
On the other hand, when restricted to the subspace $V$, the operator $\rho -
	\alpha  | \psi \rangle \langle \psi |$ is negative definite, since it has
	its all eigenvalues negative. Explicitly, $\langle {\psi ^ \bot }|\left( {\rho  - \alpha |\psi \rangle \langle \psi |} \right)|{\psi ^ \bot }\rangle  = {\left| \beta  \right|^2}{E_d} + {\left| \theta  \right|^2}{E_{d - 1}} < 0$, which is equivalent to
	\begin{equation}
	\langle {\psi ^ \bot } | \rho | {\psi ^ \bot } \rangle < 0 \ .
	\end{equation}
This conclusion contradicts the fact that $\rho$ is a density matrix, which must be positive semidefinite. There exists at most one negative eigenvalue. Finally,	 if $| \psi \rangle$ happens to be an eigenvector, then $\rho | \psi \rangle
	= \left( E + \alpha \right) | \psi \rangle$. Therefore $E \geqslant -
	\alpha$, which means $E$ can be negative, zero, or positive.
\end{proof}

\begin{lemma}[\bf Problem of eigenvalue minimization]\label{lemma:problem_eigenvalue}
Following the previous lemma, if ${E_d} \le 0$, then the minimum error $P_{\rm err}$ is minimized by minimizing the negative eigenvalue $E_d$ of the matrix ${\rho _E} - \alpha \left| \psi  \right\rangle \left\langle \psi  \right|$ with fixed $\rho_E$ and
$\alpha$.
\end{lemma}

\begin{proof}
Let us now express the minimal error as 
\begin{equation}
{P_{\rm err}} = \left( {1 - |\gamma | \cdot ||{\rho _E} - \alpha \left| \psi  \right\rangle \left\langle \psi  \right|||} \right)/2 \ ,
\end{equation}
where $\alpha  = {p_1}\eta /|\gamma |$. Denote $E_k$'s as the eigenvalues of the matrix ${\rho _E} - \alpha \left| \psi  \right\rangle \left\langle \psi  \right|$. Then, we have 
\begin{equation}
{P_{err}} = (1 - \left| \gamma  \right| \sum\nolimits_k {\left| {{E_k}} \right|} )/2   = [1 - \left| \gamma  \right| (\sum\nolimits_{k \ne d} {{E_k}}  - {E_d}) ]/2  \ ,
\end{equation}
where we have applied the result of the previous lemma. Now, we have $\sum\nolimits_{k \ne d} {{E_k}}  = 1 - \alpha  - {E_d}$ and hence 
\begin{equation}
{P_{\rm err}} = \frac{1}{2} [1 - \left| \gamma  \right|(1 - \alpha  - 2{E_d})]  \ ,
\end{equation}
which depends linearly with the smallest eigenvalue $E_d$; the more negative $E_d$ is, the smaller $P_{\rm err}$ becomes.
\end{proof}

\begin{lemma}[\bf Eigenvector for minimization]
The smallest eigenvalue of the matrix ${\rho _E} - \alpha \left| \psi  \right\rangle \left\langle \psi  \right|$, with fixed $\rho_E$ and $\alpha$, can be minimized by choosing the pure state as the eigenvector associated with the smallest eigenvalue $\lambda_{\rm min}$ of $\rho_E$, i.e., $\left| \psi  \right\rangle  = \left| {{\theta_d}} \right\rangle $ where $\left( {{\rho _E} - {\lambda _{\min }  }} \right)\left| {{\theta _d}} \right\rangle  = 0$.
\end{lemma}

\begin{proof}
First, the minimum eigenvalue ${\lambda _{\min }}\left( {A + B} \right)$ of the sum of two Hermitian matrices $A$ and $B$ is bounded by the sum of the minimum eigenvalues of the individual matrix, i.e., 
\begin{equation}
{\lambda _{\min }}\left( A \right) + {\lambda _{\min }}\left( B \right) \leqslant {\lambda _{\min }}\left( {A + B} \right) \ .
\end{equation}
As a result, the minimum eigenvalue $E_d$ of the matrix ${\rho _E} - \alpha \left| \psi  \right\rangle \left\langle \psi  \right|$ is bounded by ${\lambda _{\min }} - \alpha  \leqslant {E_d}$. Furthermore, this bound can be saturated by choosing  $\left| \psi  \right\rangle  = \left| {{\theta_d}} \right\rangle $.
\end{proof}

Therefore, for region III, we will only need to consider the minimum error resulted from sending the eigenstates of $\rho_E$ with the minimum eigenvalue. 

\begin{theorem*}[\bf Region III: minimal error decreases with reflectivity $\eta$]\label{RegionIII_class_Perr}
The minimal error $P_{\rm err}$	over all possible conventional input states is given by
\begin{equation}
P_{{\rm{err}}}^c = {p_0} + \gamma \left( {1 - {\lambda _d}} \right) \ ,
\end{equation}
which is obtained by choosing $\left| \psi  \right\rangle  = \left| {{\theta_d}} \right\rangle $ to be the eigenvector of $\rho_E$ associated with the smallest eigenvalue.
\end{theorem*}

\begin{proof}
Let us consider the matrix, 
\begin{equation}
{\Omega _{\text{c}}}(\left| {{\theta _d}} \right\rangle \langle {\theta _d}|) = {p_1}\eta \left| {{\theta _d}} \right\rangle \langle {\theta _d}| + \gamma {\rho _E} \ .
\end{equation}
Here we consider the range where 
\begin{equation}
\gamma \equiv {p_1}(1 - \eta ) - {p_0}  <0
\end{equation}
but ${p_1}\eta  + \gamma {\lambda _d} > 0$. Note that ${\text{tr(}}{\Omega _{\text{c}}}) = {p_1}\eta  + \gamma $ is a sum of the $d-1$ negative eigenvalues and one positive eigenvalue, ${p_1}\eta  + \gamma {\lambda _d}$. Therefore, the trace norm is obtained by $\left\| {{\Omega _c}} \right\| =  - {\text{tr(}}{\Omega _c}{\text{) + 2}}\left( {{p_1}\eta  + \gamma {\lambda _d}} \right) $, or
\begin{equation}
  \left\| {{\Omega _c}} \right\| = {p_1}\eta  - \gamma  + 2{\lambda _d} \gamma \ .
\end{equation}
Finally, as $P_{{\text{err}}}^c = \left( {1 - ||{\Omega _c}||} \right)/2$, we have $P_{{\text{err}}}^c = {p_0} + \gamma \left( {1 - {\lambda _d}} \right)$. 
\end{proof}

To check the consistency of the result above, when $\lambda_d = 0$, we have 
\begin{equation}
P_{{\text{err}}}^c = {p_0} + \gamma  = {p_1}\left( {1 - \eta } \right) \ .
\end{equation}
Note that when $\gamma  = 0$, then $P_{{\text{err}}}^c = {p_0}$. Moreover, when 
\begin{equation}
\eta  = (\tfrac{{{p_0}}}{{{p_1}}} - 1)(\tfrac{{{\lambda _d}}}{{1 - {\lambda _d}}}) \ ,
\end{equation}
 we have ${p_1}\eta  =  - \gamma {\lambda _d}$, which implies that $P_{{\text{err}}}^c = {p_0} + \gamma  - \gamma {\lambda _d} = {p_0} + {p_1}\eta  + ({p_1}\left( {1 - \eta } \right) - {p_0}) = {p_1}$. In the special case where $p_0=p_1=1/2$ and $d=2$, $\gamma = - \eta /2$. Therefore,
\begin{equation}
P_{{\text{err}}}^c = 1/2 - \eta /4 \ ,
\end{equation}
in agreement with the example.
%%%%%%%%%%%%%%%%%%%

\section*{Results for quantum illumination}

%%%%%%%%%%%%%%%%%%%%%%%%
\begin{lemma}[\bf Lower bound of eigenvalue value]
Given a density matrix $\rho_E$ and a pure state $\ket{\psi}$, the	minimum eigenvalue (associated with the eigenvector $\left| {{g_\psi }} \right\rangle $), ${E_g} \equiv {\lambda _{\min }}\left( {{\rho _E} \otimes {\rho _B} - \alpha \left| \psi  \right\rangle \left\langle \psi  \right|} \right)$, is bounded below by 
\begin{equation}
{E_g} \geqslant \sum\limits_{i = 1}^d {{\lambda _i}x_i^2}  - \alpha {( {\sum\limits_{i = 1}^d {{x_i}} } )^2} \ ,
\end{equation}
where ${x_i} \equiv \left| {\left\langle {{u_i}} \right.\left| {{v_i}} \right\rangle } \right|$.
\end{lemma}

\begin{proof}
Let us consider the explicit form of the smallest eigenvalue	:
\begin{equation}
  {E_g} = \left\langle g_\psi \right|{\rho _E} \otimes {\rho _B}\left| g_\psi \right\rangle  - \alpha \langle g_\psi \left| \psi  \right\rangle \left\langle \psi  \right| g_\psi \rangle \ ,
\end{equation}
where $\left\langle \psi  \right.\left| {{g_\psi }} \right\rangle  = \sum\nolimits_{i = 1}^d {\left\langle {{u_i}} \right.\left| {{v_i}} \right\rangle }$
 and $\left\langle {{g_\psi }} \right|{\rho _E} \otimes {\rho _B}\left| {{g_\psi }} \right\rangle  = \sum\nolimits_{i = 1}^d {{\lambda _i}} \left\langle {{v_i}} \right|{\rho _B}\left| {{v_i}} \right\rangle $. Now, we can obtain an upper bound of $\left\langle \psi  \right.\left| {{g_\psi }} \right\rangle $ by (i) taking absolute values for each term, i.e., 
\begin{equation}
\left\langle \psi  \right.\left| {{g_\psi }} \right\rangle  \leqslant \sum\limits_{i = 1}^d {\left| {\left\langle {{u_i}} \right.\left| {{v_i}} \right\rangle } \right|} 
 \end{equation}
(which can be achieved by choosing the vectors, $\ket{u_i}$'s and $\ket{v_i}$'s, to be proportional to each other.), and (ii) dropping all the cross terms in an expansion of $\left\langle {{v_i}} \right|{\rho _B}\left| {{v_i}} \right\rangle $, i.e., 
\begin{equation}
\left\langle {{v_i}} \right|{\rho _B}\left| {{v_i}} \right\rangle  = \sum\limits_{j = 1}^d {{{\left| {\left\langle {{u_j}} \right.\left| {{v_i}} \right\rangle } \right|}^2}}  \geqslant {\left| {\left\langle {{u_i}} \right.\left| {{v_i}} \right\rangle } \right|^2} \ .
\end{equation}
By defining ${x_i} \equiv \left| {\left\langle {{u_i}} \right.\left| {{v_i}} \right\rangle } \right|$, we obtain a lower bound for the smallest eigenvalue $E_g$, $  {E_g} \geqslant \sum\nolimits_{i = 1}^d {{\lambda _i}x_i^2}  - \alpha {( {\sum\nolimits_{i = 1}^d {{x_i}} } )^2}$.
%\begin{equation}
%  {E_g} \geqslant \sum\limits_{i = 1}^d {{\lambda _i}x_i^2}  - \alpha {( {\sum\limits_{i = 1}^d {{x_i}} } )^2} \ .
%\end{equation}
\end{proof}

Next, we are going to minimize the lower bound of $E_g$, subject to a constraint, 
\begin{equation}
\sum\limits_{i = 1}^d {{x_i}}  \equiv C \ ,
\end{equation}
where $C \leqslant 1$ is not greater than unity as $\left| {\left\langle \psi  \right.\left| {{g_\psi }} \right\rangle } \right| \leqslant 1$. We found that $E_g$ is minimized when $\alpha  >  \lambda_h$ where we have $C=1$. Here ${\lambda _h} \equiv {( {\sum\nolimits_{i = 1}^d {\lambda _i^{ - 1}} } )^{ - 1}}$ is proportional to the harmonic mean of the set of eigenvalues $\{ \lambda_i \}$ and ${\lambda _h} \leqslant {\lambda _{\min }}$. On the other hand, when $\alpha  \leqslant {\lambda_H}$, the smallest eigenvalue of $H_Q$ is positive, i.e., ${\lambda _{\min }} \geqslant 0$. In this case, the trace norm is equal to the trace, i.e., $\left\| {{H_q}} \right\| = {\text{tr}}{H_q}$. 

%%%%%%%%%%%%%%%%%%%%%%%%%
\subsection{Region II of quantum illumination}
 
\begin{lemma}[\bf Minimization with Lagrange multiplier]
The minimum eigenvalue ${E_g} = {\lambda _{\min }}\left( {{\rho _E} \otimes {\rho _B} - \alpha \left| \psi  \right\rangle \left\langle \psi  \right|} \right)$ is bounded below by a value depending on $\alpha$. (i) When $\alpha  \leqslant {\lambda _h} \equiv {(\sum\nolimits_{i = 1}^d {\lambda _i^{ - 1}} )^{ - 1}}$, $E_g \geqslant 0$. (ii) When $\alpha  > {\lambda _h}$, ${E_g} \geqslant {\lambda _h} - \alpha $.\end{lemma} 

\begin{proof}
Let us introduce a Lagrange multiplier $l_m$. The the lower bound of $E_g$, labeled by 
\begin{equation}
f \equiv\sum\limits_{i = 1}^d {{\lambda _i}x_i^2}  - \alpha {( {\sum\limits_{i = 1}^d {{x_i}} } )^2} \ ,
\end{equation}
is minimized when the condition, 
\begin{equation}
\partial f/\partial {x_i} = {l_m}\partial g/\partial {x_i} \ ,
\end{equation}
holds, where $g \equiv \sum\nolimits_{i = 1}^d {{x_i}}  - C$. Explicitly, we have
\begin{equation}
  \frac{{\partial f}}{{\partial {x_i}}} = 2({\lambda _i}{x_i} - \alpha C), \quad \frac{{\partial g}}{{\partial {x_i}}} = 1 \ ,
\end{equation}
which gives ${x_i} = \left( {\alpha C + {l_m}/2} \right)/{\lambda _i}$, and hence $C = \sum\nolimits_{i = 1}^d {{x_i}}  = \left( {\alpha C + {l_m}/2} \right)(\sum\nolimits_{i = 1}^d {\lambda _i^{ - 1}} )$. As a result, the minimal value of $f$ is given by,
\begin{equation}
  {f}  \geqslant  {C^2}({(\sum\limits_{i = 1}^d {\lambda _i^{ - 1}} )^{ - 1}} - \alpha ) \ .
\end{equation}
Therefore, When $\alpha  \leqslant {\lambda _h} \equiv {(\sum\nolimits_{i = 1}^d {\lambda _i^{ - 1}} )^{ - 1}}$, $f$ is minimized by choosing $C=0$, which implies all of the eigenvalues of $H_q$ are positive, i.e., $\lambda_i(H_q) \geqslant 0$. However, when $\alpha  > {\lambda _h}$, the lower bound $f$ is minimized by setting  $C^2=1$, which implies ${E_g} \geqslant {\lambda _h} - \alpha $.
\end{proof} 

Note that the condition of $\alpha  \leqslant {\lambda _h}$ is equivalent to 
\begin{equation}
\eta  \leqslant \tfrac{{{\lambda _h}}}{{1 - {\lambda _h}}} ( {\tfrac{{{p_0}}}{{{p_1}}} - 1} ) \ .
\end{equation}
This defines the region II of quantum illumination. There, all the eigenvalues of $\Omega_q$ has the same sign. Similar to the classical case (see theorem~\ref{theo_con_p0g12}), the minimal error is given by 
\begin{equation}
  P_{\rm err}^q = p_1 \ ,
\end{equation}
which can be achieved by any state. Recall that the region II for conventional illumination is bounded by
\begin{equation}
\eta  \leqslant ( {\tfrac{{{p_0}}}{{{p_1}}} - 1} )( {\tfrac{{{\lambda _{\min }}}}{{1 - {\lambda _{\min }}}}} ) \ ,
\end{equation}
and is equivalent to the trivial strategy. Quantum illumination is capable of shrinking the boundary to 
\begin{equation}
\eta  \leqslant \tfrac{{{\lambda _h}}}{{1 - {\lambda _h}}} ( {\tfrac{{{p_0}}}{{{p_1}}} - 1} ) \ ,
\end{equation}
as $\lambda_h \leq \lambda_{\rm min}$.
%%%%%%%%%%%%%%%%%%%%%%%%

\subsection{Region III of quantum illumination}

%\begin{lemma}[\bf Optimal state for quantum illumination]
{\bf Result: Optimal state for quantum illumination}
The lower bound, $\lambda_h - \alpha$, of $E_g$ can be achieved 	by the input state, 
\begin{equation}
\left| \psi  \right\rangle  = \sum\limits_{i = 1}^d {{\mu _i}} \left| {{\theta _i}} \right\rangle \left| {{\theta _i}} \right\rangle \ ,
\end{equation}
where ${\mu _i} = \sqrt {{\lambda _h}/{\lambda _i}} $.
%\end{lemma}

\begin{proof}
Let us consider the following ansatz, $\left| \psi  \right\rangle  = \sum\nolimits_{i = 1}^d {{\mu _i}} \left| {{\theta _i}} \right\rangle \left| {{\theta _i}} \right\rangle$,
where the amplitudes $\mu_i$'s are assumed to be non-negative ($\mu_i \geqslant 0$) and normalized ($\sum\nolimits_{i = 1}^d {\mu _i^2}  = 1$). The expectation value, $  \left\langle {{H_q}} \right\rangle  = \left\langle \psi  \right|{H_q}\left| \psi  \right\rangle = \left\langle \psi  \right|{\rho _E} \otimes {\rho _B}\left| \psi  \right\rangle - \alpha$, is given by 
\begin{equation}
  \left\langle {{H_q}} \right\rangle  = \sum\limits_{i = 1}^d {{\lambda _i}\mu _i^2}  - \alpha \ .
\end{equation}
We can achieve the lower bound, $\left\langle {{H_q}} \right\rangle  = {\lambda _h} - \alpha $,  by setting $\mu _i^2 = {\lambda _h}/{\lambda _i}$. (One can readily check that the normalization condition is obeyed as $\sum\nolimits_{i = 1}^d {\mu _i^2}  = {\lambda _h}\sum\nolimits_{i = 1}^d {1/{\lambda _i}}  = 1$.)
\end{proof}

%%%%%%%%%%%%%%%%%%%%%%%%%%%%%

\section*{Complimentary results}

%\begin{theorem*}[\bf Upper bound of minimal error]
{\bf Result: Upper bound of minimal error }
The error probability $P_{\rm err}$ is bounded above by either $p_0$ or $p_1$, i.e.,
\begin{equation}
  {P_{{\rm{err}}}} \leqslant \frac{1}{2}\left( {1 - \left| {{p_0} - {p_1}} \right|} \right) = \min \left\{ {{p_0},{p_1}} \right\} \ .
\end{equation}
%\end{theorem*}

\begin{proof}
Mathematically, this result comes from the definition of the trace norm. Let us consider the case of $p_0 \ge p_1$ first. Denote ${e_k}$'s as the eigenvalues of the operator ${p_0}{\rho _0} - {p_1}{\rho _1}$. Furthermore, we label those non-negative eigenvalues as $e_k^ + \ge 0 $ and the negative ones as $e_k^- < 0$. In this way, we can write the trace norm as the difference between these two set of eigenvalues, i.e., 
\begin{equation}
\left\| {{p_0}{\rho _0} - {p_1}{\rho _1}} \right\| = \sum\nolimits_k {e_k^ + }  - \sum\nolimits_k {e_k^ - } .
\end{equation}
In fact, using ${\rm tr}({p_0}{\rho _0} - {p_1}\rho_1 ) = {p_0} - {p_1} = \sum\nolimits_k {e_k^ + }  + \sum\nolimits_k {e_k^ - } $, we can further write $\left\| {{p_0}{\rho _0} - {p_1}{\rho _1}} \right\| = {p_0} - {p_1} - 2\sum\nolimits_k {e_k^ - }  \geqslant {p_0} - {p_1}$, which yields
\begin{equation}
{P_{\rm err}} \leqslant \tfrac{1}{2}\left( {1 - ({p_0} - {p_1})} \right) \leqslant {p_1} \ .
\end{equation}
The result for the other case, i.e., $p_0 < p_1$ can be obtained by the same argument.
\end{proof}

\begin{theorem*}[\bf Monotonicity of minimal error]
For a given reflectivity $\eta$, and density matrix $\rho$, we label the minimal error as follows, ${P_{err}}\left(\rm  \eta  \right) = \tfrac{1}{2}\left( {1 - \left\| {{p_1}{\rho _1}\left( \eta  \right) - {p_0}{\rho _0}} \right\|} \right)$, where ${\rho _1}\left( \eta  \right) = \eta \rho  + \left( {1 - \eta } \right){\rho _0}$. The minimal error is a non-increasing function of the reflectivity, i.e., if $\eta  \geqslant \eta '$, then
\begin{equation}
  {P_{\rm err}}\left( \eta  \right) \leqslant {P_{\rm err}}\left( {\eta '} \right) \ .
\end{equation}
\end{theorem*}

\begin{proof}
First of all, we define a trace-preserving quantum operation that represents the action of the reflection, 
\begin{equation}
{\varepsilon _\eta }\left( \rho  \right) = \eta \rho  + \left( {1 - \eta } \right){\rho_0} \ ,
\end{equation}
and hence ${\varepsilon _\eta }\left( {{\rho_0}} \right) = {\rho_0}$. Of course, there exist a quantum operation, 
\begin{equation}
{\varepsilon _\Delta }\left( \rho  \right) = \Delta \rho  + \left( {1 - \Delta } \right){\rho_0} \ ,
\end{equation}
where $\Delta \eta  = \eta '$, connecting the two, i.e., 
\begin{equation}
{\varepsilon _\Delta }\left( {{\varepsilon _\eta }\left( \rho  \right)} \right) = {\varepsilon _{\eta '}}\left( \rho  \right) \ .
\end{equation}
In this way, all we need to show is that 
\begin{equation}
\left\| {{p_1}{\varepsilon _{\eta '}}\left( \rho  \right) - {p_0}{\varepsilon _{\eta '}}\left( {{\rho_0}} \right)} \right\| \leqslant \left\| {{p_1}{\varepsilon _\eta }\left( \rho  \right) - {p_0}{\varepsilon _\eta }\left( {{\rho_0}} \right)} \right\| \ ,
\end{equation}
or equivalently, 
\begin{equation}
\left\| {{p_1}{\varepsilon _\Delta }\left( \rho  \right) - {p_0}{\varepsilon _\Delta }\left( {{\rho_0}} \right)} \right\| \leqslant \left\| {{p_1}\rho  - {p_0}{\rho_0}} \right\| \ ,
\end{equation}
for any density matrix $\rho$. The proof for the latter inequality is essentially the same as the well-known result that a trace-preserving operation is contractive under the measure of trace distance. 

Alternatively, one can prove it by contradiction as follows: first, the Helstrom bound states that the minimal error for distinguishing between density matrices $\rho$ and $\rho_0$ (occurring with probabilities $p_0$ and $p_1$ respectively) is determined by the trace norm $\left\| {{p_1}\rho  - {p_0}{\rho_0}} \right\|$, over all possible POVM measurements.  Since any channel can be regarded as a type of POVM, if the opposite, i.e., 
\begin{equation}
\left\| {{p_1}{\varepsilon _\Delta }\left( \rho  \right) - {p_0}{\varepsilon _\Delta }\left( {{\rho_0}} \right)} \right\| \geqslant \left\| {{p_1}\rho  - {p_0}{\rho_0}} \right\| \ ,
\end{equation}
were true, then one could find a POVM that yields a smaller error than the Helstrom bound, which causes a contradiction.
\end{proof}

\begin{theorem*}[\bf Optimization over pure states]\label{theo:optim_purestate}
In optimizing of the trace norm of the matrices $\Omega_{c} (\rho)$ or $\Omega_{q} (\rho)$ over all possible density matrices, we only need to maximize the set of pure states. i.e.,
\begin{equation}\label{max_rho_2_max_pure}
\mathop {\max }\limits_{\rho  \in \mathcal{H}} \left\| {{ \, \Omega _{\text{c}}}  (\rho )} \, \right\| = \mathop {\max }\limits_{|\psi \rangle  \in \mathcal{H}} \left\| {{ \, \Omega _{\text{c}} }\, (\left| \psi  \right\rangle \left\langle \psi  \right|)} \, \right\|,
\end{equation}
\begin{equation}
\mathop {\max }\limits_{\rho_{AB}  \in \mathcal{H}\otimes\mathcal{H}} \left\| {{ \, \Omega _{\text{q}}}  (\rho_{AB} )} \, \right\| = \mathop {\max }\limits_{|\psi_{AB} \rangle  \in \mathcal{H}\otimes\mathcal{H}} \left\| {{ \, \Omega _{\text{q}} }\, (\left| \psi_{AB}  \right\rangle \left\langle \psi_{AB}  \right|)} \, \right\| .
\end{equation}
\end{theorem*}

\begin{proof}
Let us denote, 
\begin{equation}
{\rho _*} \equiv \sum\nolimits_i {{\lambda _i}|{\psi _i}\rangle \langle {\psi _i}|} \ ,
\end{equation}
expressed in some diagonal basis $\left\{ {\left| {{\psi _i}} \right\rangle } \right\}$ where all ${\lambda _i} \geqslant 0$ and $\sum\nolimits_i {{\lambda _i}}  = 1$, as the density matrix that maximizes the trace norm of ${\Omega _c}$, i.e., 
\begin{equation}
\left\| {{\Omega _{\rm c}}\left( {{\rho _*}} \right)} \right\| = \mathop {\max }\limits_{\rho  \in {\cal H}} \left\| {{\Omega _{\rm c}}\left( \rho  \right)} \right\| \ .
\end{equation}
Since ${\Omega _c}$ is linear, which implies that 
\begin{equation}
{\Omega _c}\left( {\sum\nolimits_i {{\lambda _i}|{\psi _i}\rangle \langle {\psi _i}|} } \right) = \sum\nolimits_i {{\lambda _i}{\Omega _c}\left( {|{\psi _i}\rangle \langle {\psi _i}|} \right)} \ ,
\end{equation}
and that the trace norm is convex, i.e., 
\begin{equation}
\left\| {\sum\nolimits_i {{\lambda _i}{\Omega _c}\left( {|{\psi _i}\rangle \langle {\psi _i}|} \right)} } \right\| \leqslant \sum\nolimits_i {{\lambda _i}\left\| {{\Omega _c}\left( {|{\psi _i}\rangle \langle {\psi _i}|} \right)} \right\|} \ .
\end{equation}
We have $\left\| {{\Omega _c}\left( {{\rho _*}} \right)} \right\| \leqslant \sum\nolimits_i {{\lambda _i}\left\| {{\Omega _c}\left( {|{\psi _i}\rangle \langle {\psi _i}|} \right)} \right\|}$,
which implies the result advertised in Eq.~(\ref{max_rho_2_max_pure}). The case for ${\Omega _q}$ is similar.
\end{proof}

\end{document}